\documentclass[journal,letterpaper,onecolumn, 11pt, draftclsnofoot]{IEEEtran}
\usepackage{amsmath,epsfig,amsfonts,amssymb,mathrsfs,pifont,multirow,fmtcount}

\newtheorem{theorem}{Theorem}
\newtheorem{lemma}{Lemma}
\newtheorem{corollary}{Corollary}

\usepackage{cite}

\ifCLASSINFOpdf
\else
\fi
\usepackage{xspace}
\usepackage{bbm}
\usepackage{mathrsfs}

\def\diag{\mathop{\rm diag}\nolimits}%
\def\rank{\mathop{\rm rank}\nolimits}%
\newcommand{\Ac}{\mathcal{A}}
\newcommand{\Bc}{\mathcal{B}}

\newcommand{\Ec}{\mathcal{E}}

\newcommand{\Qc}{\mathcal{Q}}

\newcommand{\Sc}{\mathcal{S}}
\newcommand{\Tc}{\mathcal{T}}

\newcommand{\Xv}{{\bf X}}
\newcommand{\Av}{{\bf A}}
\newcommand{\av}{{\bf a}}

\newcommand{\Dv}{{\bf D}}
\newcommand{\Yv}{{\bf Y}}
\newcommand{\Zv}{{\bf Z}}
\newcommand{\Wv}{{\bf W}}

\newcommand{\Sv}{{\bf S}}

\newcommand{\hv}{{\bf h}}

\newcommand{\wv}{{\bf w}}
\newcommand{\xv}{{\bf x}}

\newcommand{\bv}{{\bf b}}

\DeclareMathOperator\E{\sf E}
\let\P\relax
\DeclareMathOperator\P{\sf P}

\def\textiid{i.i.d.\@\xspace}
\newcommand\iid{\ifmmode\text{ i.i.d. } \else \textiid \fi}

\newcommand{\supp}{{\textrm{supp}}}
\begin{document}
%
\title{Support Recovery of Sparse Signals in the Presence of Multiple Measurement Vectors}
%
%
%

\author{Yuzhe~Jin,~\IEEEmembership{Student Member,~IEEE,}
        and~Bhaskar~D.~Rao,~\IEEEmembership{Fellow,~IEEE}\thanks{The material in this paper was presented in part at the Asilomar Conference on Signals, Systems, and Computers, Pacific Grove, California, USA, 2010.}}
\maketitle

\begin{abstract}
This paper studies the problem of support recovery of sparse signals based on multiple measurement vectors (MMV). The MMV support recovery problem is connected to the problem of decoding messages in a Single-Input Multiple-Output (SIMO) multiple access channel (MAC), thereby enabling an information theoretic framework for analyzing performance limits in recovering the support of sparse signals. Sharp sufficient and necessary conditions for successful support recovery are derived in terms of the number of measurements per measurement vector, the number of nonzero rows, the measurement noise level, and especially the number of measurement vectors. Through the interpretations of the results, in particular the connection to the multiple output communication system, the benefit of having MMV for sparse signal recovery is illustrated providing a theoretical foundation to the performance improvement enabled by MMV as observed in many existing simulation results. In particular, it is shown that the structure (rank) of the matrix formed by the nonzero entries plays an important role on the performance limits of support recovery.
\end{abstract}

\section{Introduction}
Suppose the signal of interest is $X\in\mathbb{R}^{m\times l}$, and $X$ is said to be sparse when only a few of its rows contain nonzero elements whereas the rest consist of zero elements. One wishes to estimate $X$ via the linear measurements $Y = AX + Z$, where $A\in\mathbb{R}^{n\times m}$ is the measurement matrix and $Z\in\mathbb{R}^{n\times l}$ is the measurement noise. The goal is to estimate $X$ from as few measurements as possible. Specifically, when $l=1$, this problem is usually termed as sparse signal recovery with  a single measurement vector (SMV); when $l>1$, it is  referred to as sparse signal recovery with multiple measurement vectors (MMV) \cite{Rao05,Cichocki}. This problem has received much attention in many disciplines motivated by a broad array of applications such as compressed sensing \cite{Donoho06_4,Candes_comp}, biomagnetic inverse problems \cite{Rao97}, \cite{Rao98}, image processing \cite{Jeffs98}, \cite{Baraniuksinglepixel}, robust face recognition \cite{WrightMa}, bandlimited extrapolation and spectral estimation \cite{Parks91}, robust regression and outlier detection \cite{JR_10ICASSP}, speech processing \cite{ChuSpeech}, channel estimation \cite{RaoCotter02}, \cite{NowakChannel}, echo cancellation \cite{Duttweiler2000}, \cite{RaoSong}, body area networks \cite{garudadri}, and wireless communication \cite{RaoCotter02}, \cite{Guo09ITA}.

\subsection{Background on the SMV Problem}
For the problem of sparse signal recovery with SMV, computationally efficient algorithms have been proposed to find or approximate the sparse solution $X \in\mathbb{R}^{m}$ in various settings. A partial list includes matching pursuit \cite{Zhang93}, orthogonal matching pursuit (OMP) \cite{Pati93}, Lasso \cite{Tibshirani94}, basis pursuit \cite{Donoho01}, FOCUSS \cite{Rao97}, iteratively reweighted $\ell_1$ minimization \cite{Candes07enhancingsparsity}, iteratively reweighted $\ell_2$ minimization \cite{ChartrandYin}, sparse Bayesian learning (SBL) \cite{Tipping01,Wipf_2004}, finite rate of innovation \cite{Vetterli02}, CoSaMP\cite{NeedellTropp08}, and subspace pursuit \cite{DaiMilenkovic08}. Analysis has been developed to shed light on the performances of these practical algorithms. For example, Donoho \cite{Donoho06_4}, Donoho, Elad, and Temlyakov \cite{Donoho06_1}, Cand\`{e}s and Tao \cite{Tao05a}, and Cand\`{e}s, Romberg, and Tao\cite{Tao06} presented sufficient conditions for $\ell_1$-norm minimization algorithms, including basis pursuit and its variant in the noisy setting, to successfully recover the sparse signals with respect to different performance metrics. Tropp \cite{Tropp04}, Tropp and Gilbert \cite{Tropp07}, and Donoho, Tsaig, Drori, and Starck \cite{Donoho06} studied the performances of greedy sequential selection methods such as matching pursuit and its variants. Wainwright \cite{Wainwright09ell1} and Zhao and Yu \cite{ZhaoYu2006} provided sufficient and necessary conditions for Lasso to recover the support of the sparse signal, i.e., the set of indices of the nonzero entries. On the other hand, from an information theoretic perspective, a series of papers, for instance, Wainwright \cite{Wain09b},  Fletcher, Rangan, and Goyal \cite{Fletcher09IT}, Wang, Wainwright, and Ramchandran \cite{Wang08ISIT}, Ak\c{c}akaya and Tarokh \cite{Akcakaya}, Jin, Kim, and Rao \cite{JKR10}, provided sufficient and necessary conditions to characterize the performance limits of optimal algorithms for support recovery, regardless of computational complexity.

\subsection{Background on the MMV Problem}
As a fast emerging trend, the capability of collecting multiple measurements with an array of sensors in an increasing number of applications, such as magnetoencephalography (MEG) and electroencephalography (EEG) \cite{WipfNaga2008, Cichocki}, blind source separation \cite{Cichocki_2}, multivariate regression \cite{OWainJ}, and source localization \cite{Malioutov_sourceLocalization}, gives rise to the problem of sparse signal recovery with multiple measurement vectors. Practical algorithms have been developed to address the new challenges in this scenario. One class of algorithms for solving the MMV problem can be viewed as straightforward extensions based on their counterparts in the SMV problem. To sample a few, M-OMP \cite{RaoCotter02, Tropp_05_ICASSP}, M-FOCUSS \cite{RaoCotter02}, $\ell_1/\ell_2$ minimization method\footnote{This method is sometimes referred to as $\ell_2/\ell_1$ minimization, due to the naming convention in a specific paper. In this paper, we use $\ell_1/\ell_p$ to indicate a cost of a matrix $B$ which is define as $\sum_{i}|(\sum_j|b_{i,j}|^p)^{1/p}|$.} \cite{Eldar_UnionSubspace}, multivariate group Lasso \cite{OWainJ}, and M-SBL \cite{Wipf_2007_b} can be all viewed as examples of this kind. Another class of algorithms additionally make explicit effort to exploit the structure underlying the sparse signal $X$, such as the temporal correlation or the autoregressive nature across the columns of $X$ which would be otherwise unavailable when $l=1$, to aim for better performance of sparse signal recovery. For instance, the improved M-FOCUSS algorithms \cite{Cichocki} and the auto-regressive sparse Bayesian learning (AR-SBL) \cite{ZhilinRao} both have the capability of explicitly taking advantage of the structural properties of $X$ to improve the recovery performance. Along side the algorithmic advancement, a series of work have been focusing on the theoretical analysis to support the effectiveness of existing algorithms for the MMV problem. We briefly divide these results into two categories. The first category of theoretic analysis aims at specific practical algorithms for sparse signal recovery with MMV. For example, Chen and Huo \cite{ChenHuo} discovered the sufficient conditions for $\ell_1/\ell_p$ norm minimization method and orthogonal matching pursuit to exactly recover every sparse signal within certain sparsity level in the noiseless setting. Eldar and Rauhut \cite{EldarRauhut} also analyzed the performance of sparse recovery using the $\ell_1/\ell_2$ norm minimization method in the noiseless setting, but the sparse signal was assumed to be randomly distributed according to certain probability distribution and the performance was averaged over all possible realizations of the sparse signal. Obozinski, Wainwright, and Jordan \cite{OWainJ} provided sufficient and necessary conditions for multivariate group Lasso to successfully recover the support of the sparse signal\footnote{We refer to the support of a matrix $X$ as the set of indices corresponding to the nonzero rows of $X$. It will be formally defined in Section \ref{sec:ProblemFormulation}.}  in the presence of measurement noise. The second category of theoretic analysis are of an information theoretic nature, and  explore the performance limits that any algorithm, regardless of computational complexity could possibly achieve. In this regard, Tang and Nehorai \cite{TangNehorai} employed a hypothesis testing framework with the likelihood ratio test as the optimal decision rule to study how fast the error probability decays. Sufficient and necessary conditions are further identified in order to guarantee successful support recovery in the asymptotic sense.

\subsection{Focus and Contributions of This Paper}
We develop sharp asymptotic performance limits among the signal dimension $m$, the number of nonzero rows $k$, the number of measurements per measurement vector $n$, and the number of measurement vectors $l$ for reliable support recovery in the noisy setting. We show that $n = (\log m)/c(X)$ is sufficient and necessary. We give a complete characterization of $c(X)$ that depends on the elements of the nonzero rows of $X$. Together with interpretations, we demonstrate the potential performance improvement enabled by having MMV, and hence bolster its usage in practical applications. Our main results are inspired by the analogy to communication over a Single-Input Multiple-Output (SIMO) multiple access channel (MAC). According to this connection, the columns of the measurement matrix form a common codebook for all senders. Codewords from the senders are individually multiplied by unknown channel gains, which correspond to nonzero entries of $X$. Then, the noise corrupted linear combinations of these codewords are observed by multiple receivers, which correspond to the multiple measurement vectors. The problem can be viewed as $k$ single-antenna users communicating over a non-frequency selective channel with a base station equipped with $l$ receive antennas. Thus, the problem of support recovery can be interpreted as multiple receivers jointly decoding messages sent by multiple senders, i.e., a SIMO MAC channel. With appropriate modifications, the techniques for deriving multiple-user channel capacity can be leveraged to provide performance limits for support recovery.

In the literatures on sparse signal recovery with SMV, the analogy between the problems of sparse signal recovery and channel coding has been observed from various perspectives in previous work \cite{Bara06}, \cite[Section IV-D]{Tropp06}, \cite[Section II-A]{Wang08ISIT}, \cite[Section III-A]{Akcakaya}, \cite[Section 11.2]{Donoho06}. However, their extensions to the MMV problem are unavailable to the authors' knowledge. Moreover, our approach differs from existing works and would be different form their possible extensions to the MMV scenario, if any. We customize tools from multiple-user information theory to address the support recovery problem and we obtain sharp performance limits in the form of tight sufficient and necessary conditions.

\subsection{Organization of the Paper}
In Section \ref{sec:ProblemFormulation}, we formally define the problem of support recovery of sparse signals in the presence of MMV. To motivate the main results of the paper and their proof techniques, in Section \ref{sec:ConnectionMACMultipleReceiver} we discuss the similarities and differences between the support recovery problem and the multiple access communication problem. The main results of the paper are presented in Section \ref{sec:mainresult}, along with the interpretations. The proofs of the main theorems are presented in Appendices \ref{app:th1} and \ref{app:th2}. Relations to existing work are discussed in Section \ref{sec:relation}. Section \ref{sec:discussion} concludes the paper with further discussions.

\subsection{Notations}
Throughout this paper, a set is a collection of unique objects. Let $\mathbb{R}^m$ denote the $m$-dimensional real
Euclidean space. Let $\mathbf{1}$ denote a column vector whose elements are all $1$'s, and its length can be determined in the context. Let $\mathbb{N}=\{1, 2, 3, ...\}$ denote the set of
natural numbers. Let $[k]$ denote the set $\{1, 2, ...,
k\}$. The notation $|\Sc|$ denotes the cardinality of set $\Sc$,
$\|\xv\|_2$ denotes the $\ell_2$-norm of a vector $\xv$,
and $\|A\|_F$ denotes the Frobenius norm of a matrix $A$. For a matrix $A$, $\Av_i$ denotes its $i$th column, $\underline{A}_i$ denotes its $i$th row, and $\underline{A}_\Tc$ denotes the submatrix formed by the rows of $A$ indexed by the set $\Tc$.

\section{Problem Formulation}
\label{sec:ProblemFormulation}

Let $W\in\mathbb{R}^{k\times l}$, where $w_{i,j}\neq 0$ for all $i, j$. Let $\Sv=[S_1, ..., S_k]^\intercal\in[m]^k$ be such that $S_1$, ..., $S_k$  are chosen uniformly at random from $[m]$ without replacement. In particular, $\{S_1, ...,S_k\}$ is uniformly distributed over all size-$k$ subsets of $[m]$. Then, the signal of interest $X=X(W,\Sv)$ is generated as
\begin{align}
X_{s,i} = \left\{ \begin{array}{ll}
w_{j,i} & \mbox{if $s=S_j$},\\
0 & \mbox{if $s\notin \{S_1, ..., S_k\}$}.\end{array} \right.
\label{signal_model}
\end{align}
The support of $X$, denoted by $\supp(X)$, is the set of indices corresponding to the nonzero rows of $X$, i.e., $\supp(X) = \{S_1,...,S_k\}$. According to the signal model (\ref{signal_model}), $|\supp(X)|= k$. Throughout this paper, we assume $k$ is known.

We measure $X$ through the linear operation
\begin{align}
Y = A X + Z
\label{CS_model}
\end{align}
where $A\in \mathbb{R}^{n\times m}$ is the measurement matrix, $Z\in\mathbb{R}^{n\times l}$ is the measurement noise, and $Y\in\mathbb{R}^{n\times l}$ is the noisy measurement. We assume that the elements of $A$ are independent and identically distributed (i.i.d.) according to the Gaussian distribution $\mathcal{N}(0, \sigma_a^2)$, and the noise $Z_{i,j}$ are  i.i.d. according to  $\mathcal{N}(0, \sigma_z^2)$. We assume $\sigma_a^2$ and $\sigma_z^2$ are known.

Upon observing the noisy measurement $Y$, the goal is to recover the indices of the nonzero rows of $X$. A support recovery map is defined as
\begin{align}
d: \mathbb{R}^{n\times l} \longmapsto 2^{[m]}.
\label{recon_algorithm}
\end{align}

Given the signal model (\ref{signal_model}), the measurement model (\ref{CS_model}), and the support recovery map (\ref{recon_algorithm}), we define the average probability of error by
\[
\P\{d(Y)\neq \textmd{supp}(X(W, \Sv))\}
\]
for each (unknown) signal value matrix $W\in\mathbb{R}^{k\times l}$. Note that the probability is averaged over the randomness of the locations of the nonzero rows $\Sv$, the measurement matrix $A$, and the measurement noise $Z$.

\section{Interpretation of Support Recovery via Multiple-User Communication}
\label{sec:ConnectionMACMultipleReceiver}
We introduce an important interpretation of the problem of support recovery of sparse signals by relating it to a single-input multiple-output (SIMO) multiple access channel (MAC) communication problem. This relationship motivates the intuition behind our main results and facilities the development of the proof techniques. It can be also viewed as an MMV extension of our earlier work \cite{JKRIT10}, in which a similar connection was employed to interpret the support recovery problem with SMV.

\subsection{Brief Review on SIMO MAC}
Consider the following wireless communication scenario. Suppose $k$ senders wish to transmit information to a set of $l$ common receivers. Each sender $i$ has access to a codebook $\mathscr{C}^{(i)} = \{\mathbf{c}_1^{(i)}, \mathbf{c}_2^{(i)},..., \mathbf{c}_{m^{(i)}}^{(i)}\}$, where $\mathbf{c}_j^{(i)}\in\mathbb{R}^n$ is a codeword and $m^{(i)}$ is the number of codewords in the codebook. The rate for sender $i$ is $R^{(i)} = (\log m^{(i)})/n$. To transmit information, each sender chooses a codeword from its codebook, and all senders transmit their codewords simultaneously to  $l$ receivers leading to the SIMO MAC problem:
\begin{align}
Y_{j,i} =  h_{j,1} X_{1,i} + h_{j,2} X_{2,i}  + \cdots +h_{j,k} X_{k,i} + Z_{j,i},\quad i=1, 2,...,n\textrm{, and }j =1, 2, ...,l
\label{MAC_channel}
\end{align}
where $X_{q, i}$ denotes the input symbol from sender $q$ to the channel at the $i$th use of the channel, $h_{j, q}$ denotes the channel gain between sender $q$ and receiver $j$, $Z_{j, i}$ is the additive Gaussian noise i.i.d. according to $\mathcal{N}(0, \sigma_z^2)$, and $Y_{j, i}$ is the channel output at receiver $j$ at the $i$th use of the channel.

After receiving $Y_{j,1,}, ...,Y_{j,n}$ at each receiver $j\in[l]$, the receivers work jointly to determine the codewords transmitted by each sender. Since the senders interfere with each other, there is an inherent tradeoff among their operating rates. The notion of capacity region is introduced to capture this tradeoff by characterizing all possible rate tuples $(R^{(1)}, R^{(2)}, ..., R^{(k)})$ at which reliable communication can be achieved with diminishing error probability of decoding. By assuming each sender obeys the power constraint $
\|\mathbf{c}_j^{(i)}\|^2/n\leq \sigma_c^2
$
for all $j\in[m^{(i)}]$ and all $i\in[k]$, the capacity region of a SIMO MAC with known channel gains \cite{GoldsmithCapacity} is
\begin{align}
\left\{(R^{(1)},...,R^{(k)}):\sum_{i\in \Tc} R^{(i)} \leq \frac{1}{2}\log \left(I+\frac{\sigma_c^2}{\sigma_z^2}\sum\limits_{i\in \Tc} \hv_i\hv_i^\intercal
\right), \forall ~\Tc \subseteq [k]\right\}
\label{Kusecaparegion}
\end{align}
where $\hv_i \triangleq [h_{1, i}, ..., h_{l, i}]^\intercal$ for $i\in [k]$.

\subsection{Similarities and Differences to the Problem of Support Recovery}
Based on the measurement model (\ref{CS_model}), we can remove the columns in $A$ which correspond to the zero rows of $X$, and obtain the following effective form of the measurement procedure
\begin{align}
\Yv_j =  X_{S_1, j} \Av_{S_1} + \cdots + X_{S_k, j} \Av_{S_k} + \Zv_{j}
\label{CS_peeloff}
\end{align}
for $j\in[l]$. By contrasting (\ref{CS_peeloff}) to the SIMO MAC (\ref{MAC_channel}), we can draw the following key connections that relate the two problems \cite{Jin08}.
\begin{enumerate}
\item[i)] {\bf A nonzero entry as a sender}: We can view the existence of a nonzero row index $S_i$ as sender $i$ that accesses the channel. Since there are $k$ nonzero entries, this results in $k$ users leading to the MAC analogy.

\item[ii)] {\bf A measurement vector as a receiver}: We can view the existence of a measurement vector $\Yv_j$ as a measurement at receiver $j$. The multiple receivers leads to the multiple output (MO) part of the analogy.

\item[iii)] {\bf $X_{S_i, j}$ as the channel gain}: The nonzero entry $X_{S_i, j}$, i.e., $w_{i, j}$, plays the role of the channel gain $h_{j, i}$ from the $i$th sender to the $j$th receiver.

\item[iv)]  {\bf  $\Av_i$ as the codeword}: We treat the measurement matrix $A$ as a codebook with each column
$\Av_i$, $i\in [m]$, as a codeword. Each element of $\Av_{S_i}$ is fed one by one through the channel as input symbols for the $i$th sender to the $l$ receivers, resulting in $n$ uses of the channel. Since a users transmits a single stream, this leads to
the single input (SI) part of the analogy.

\item[v)] {\bf Similarity of objectives}: In the problem of sparse signal recovery, we focus on finding the support $\{S_1,...,S_k\}$ of the signal. In the problem of MAC communication, the receiver needs to determine the indices of codewords, i.e., $S_1,...,S_k$, that are transmitted by the senders.
\end{enumerate}

Based on the abovementioned aspects, the two problems share significant similarities which enable leveraging the information theoretic methods for the SIMO MAC problem for the performance analysis of support recovery of sparse signals. However, there are domain specific differences between the support recovery problem and the channel coding problem that should be addressed accordingly to rigorously apply the information theoretic approaches \cite{JKRIT10}.

\begin{enumerate}
\item\textbf{Common codebook}: In MAC communication, each sender uses its own codebook. However, in sparse signal recovery, the ``codebook'' $A$ is shared by all ``senders''. All senders choose their codewords from the same codebook and hence operate at the same rate. Different senders will not choose the same codeword, or they will collapse into one sender.

\item\textbf{Unknown channel gains}: In MAC communication, the capacity region (\ref{Kusecaparegion}) is valid assuming that the receiver knows the channel gain $h_i$ \cite{Tse05}. In contrast, for sparse signal recovery problem, $X_{S_i}$ is actually unknown and needs to be estimated. Although coding techniques and capacity results are available for communication with channel uncertainty, a closer examination indicates that those results are not directly applicable to our problem. For instance, channel training with pilot symbols is a common practice to combat channel uncertainty \cite{Hassibi00howmuch}. However, it is not obvious how to incorporate the training procedure into the measurement model (\ref{CS_model}), and hence the related results are not directly applicable.
\end{enumerate}

Once these differences are properly accounted for, the connection between the problems of sparse signal recovery and channel coding makes available a variety of information theoretic tools for handling performance issues pertaining to the support recovery problem. Based on techniques that are rooted in channel capacity results, but suitably modified to deal with the differences, we present the main results of this paper in the next section.

\section{Main Results and Their Interpretations}
\label{sec:mainresult}
\subsection{Main Results}
We consider the support recovery of a sequence of sparse signals generated with the same signal value matrix $W$. In particular, we assume that $k$ and $l$ are fixed. Define the auxiliary quantity
\begin{align}
c(W)\triangleq \min_{\Tc\subseteq [k]}\left[\frac{1}{2|\Tc|}\log\det\left(I + \frac{\sigma_a^2}{\sigma_z^2}\underline{W}_{\Tc}^\intercal \underline{W}_{\Tc}\right)  \right].
\end{align}
The following two theorems summarize the main results. The proofs are presented in Appendices A and B.

\begin{theorem}
\label{theorem1}
If
\begin{align}
\limsup_{m\rightarrow\infty}\frac{\log m}{n_m} < c(W)
\end{align}
then there exists a sequence of support recovery maps $\{d^{(m)}\}_{m=k}^\infty,d^{(m)}:\mathbb{R}^{n_m\times l}\mapsto 2^{[m]}$, such that
\begin{align}
\lim_{m\rightarrow \infty}\P\{d(Y)\neq \textmd{supp}(X(W, \Sv))\} =0.
\end{align}
\end{theorem}

\begin{theorem}
If
\begin{align}
\limsup_{m\rightarrow\infty}\frac{\log m}{n_m} > c(W)
\end{align}
then for any sequence of support recovery maps $\{d^{(m)}\}_{m=k}^\infty,d^{(m)}:\mathbb{R}^{n_m\times l}\mapsto 2^{[m]}$, \begin{align}
\liminf_{m\rightarrow \infty}\P\{d(Y)\neq \textmd{supp}(X(W, \Sv))\} >0.
\end{align}
\end{theorem}

Theorems 1 and 2 together indicate that $n = \frac{1}{c(W)\pm \epsilon}\log m$ is the sufficient and necessary number of measurements per measurement vector to ensure asymptotically successful support recovery. The constant $c(W)$ explicitly captures the role of the nonzero entries in the performance tradeoff.

\subsection{Interpretations of the Main Results}
We further explore the implications of having multiple measurement vectors. Due to the complicated nature of the expression for $c(W)$, we will employ different approximations to make the interpretations more accessible.

\medskip
\subsubsection{The Low-Noise-Level Scenario}
We consider the case where ${\sigma_z^2}$ is sufficiently small. Let $\lambda_{\Tc, i},\Tc\subseteq[k]$, denote the $i$th largest eigenvalue of $\underline{W}_{\Tc}^\intercal \underline{W}_{\Tc}$. For a SIMO MAC problem, the sum capacity grows as $\min(k, l)$ leading to significant gains in the task of support recovery. This is captured in the following corollary.
\begin{corollary}
\label{corollary1}
For a given $W$, suppose $\rank(\underline{W}_{\Tc}^\intercal \underline{W}_{\Tc})=\min(|\Tc|, l)$ for all $\Tc\subseteq[k]$. For sufficiently small $\sigma_z^2>0$, there exists a constant $\alpha\in(0, 1)$ such that if
\begin{align}
\lim_{ m\rightarrow \infty}\frac{\log m}{n_m} &< \alpha \cdot \frac{\min(k,l)}{2k}\cdot \log \frac{\sigma_a^2}{\sigma_z^2}\label{suffCor1}
\end{align}
then there exists a sequence of support recovery maps $\{d^{(m)}\}_{m=k}^\infty,d^{(m)}:\mathbb{R}^{n_m\times l}\mapsto 2^{[m]}$, such that
\begin{align}
\lim_{m\rightarrow \infty}\P\{d(Y)\neq \textmd{supp}(X(W, \Sv))\} =0.\nonumber
\end{align}
\end{corollary}

\begin{proof}
Note that for $\Tc\subseteq[k]$ with $|\Tc|\leq l$, $\lambda_{\Tc, i} >0$ for $i=1, 2, ..., |\Tc|$. Thus
\begin{align}
\frac{1}{2|\Tc|}\log\det\left(I + \frac{\sigma_a^2}{\sigma_z^2}\underline{W}_{\Tc}^\intercal \underline{W}_{\Tc}\right)&= \frac{1}{2|\Tc|}\log \prod_{i=1}^{|\Tc|}\left(1+\frac{\sigma_a^2}{\sigma_z^2}\lambda_{\Tc, i}\right)\nonumber\\
&= \frac{1}{2|\Tc|}\log \prod_{i=1}^{|\Tc|}\left(\frac{\sigma_a^2}{\sigma_z^2}\left(\frac{\sigma_z^2}{\sigma_a^2}+\lambda_{\Tc, i}\right)\right)\nonumber\\
&=  \frac{1}{2|\Tc|}\left[{|\Tc|}\cdot \log \frac{\sigma_a^2}{\sigma_z^2} + \sum_{i=1}^{|\Tc|}\log \left(\frac{\sigma_z^2}{\sigma_a^2}+\lambda_{\Tc, i}\right)\right]\nonumber\\
&= \frac{1}{2} \log \frac{\sigma_a^2}{\sigma_z^2}\cdot\left(1 + \frac{1}{|\Tc|}\sum_{i=1}^{|\Tc|}\frac{\log \left(\frac{\sigma_z^2}{\sigma_a^2}+\lambda_{\Tc, i}\right)}{\log \frac{\sigma_a^2}{\sigma_z^2}}\right)\nonumber\\
&= \frac{1}{2} \log \frac{\sigma_a^2}{\sigma_z^2}\cdot\left(1 + O\left(\frac{1}{-\log \sigma_z^2}\right)\right)\nonumber\\
&\geq \frac{1}{2} \log \frac{\sigma_a^2}{\sigma_z^2}\cdot \alpha_\Tc\nonumber
\end{align}
for some $\alpha_\Tc\in(0, 1)$. For any possible $\Tc\subseteq[k]$ with $|\Tc| > l$, $\lambda_{\Tc, i} >0$ for $i=1, 2, ..., l$. Then, we have similarly
\begin{align}
\frac{1}{2|\Tc|}\log\det\left(I + \frac{\sigma_a^2}{\sigma_z^2}\underline{W}_{\Tc}^\intercal \underline{W}_{\Tc}\right)
&=  \frac{l}{2|\Tc|} \log \frac{\sigma_a^2}{\sigma_z^2}\cdot\left(1 + O\left(\frac{1}{-\log \sigma_z^2}\right)\right).\nonumber\\
&\geq \frac{l}{2|\Tc|} \log \frac{\sigma_a^2}{\sigma_z^2}\cdot \alpha_\Tc.\nonumber
\end{align}
Thus, if $k\leq l$
\begin{align}
\min_{\Tc\subseteq [k]}\left[\frac{1}{2|\Tc|}\log\det\left(I + \frac{\sigma_a^2}{\sigma_z^2}\underline{W}_{\Tc}^\intercal \underline{W}_{\Tc}\right)  \right] \geq \frac{1}{2} \log \frac{\sigma_a^2}{\sigma_z^2}\cdot\min_{\Tc\subseteq [k]}\alpha_\Tc\label{llarger}
\end{align}
and if $k>l$
\begin{align}
\min_{\Tc\subseteq [k]}\left[\frac{1}{2|\Tc|}\log\det\left(I + \frac{\sigma_a^2}{\sigma_z^2}\underline{W}_{\Tc}^\intercal \underline{W}_{\Tc}\right)  \right] \geq \frac{l}{2k} \log \frac{\sigma_a^2}{\sigma_z^2}\cdot\min_{\Tc\subseteq [k]}\alpha_\Tc.\label{klarger}
\end{align}
Combining (\ref{llarger}) and (\ref{klarger}) and applying Theorem \ref{theorem1} complete the proof.
\end{proof}
\medskip

Corollary \ref{corollary1} indicates the following observations. First, as the measurement noise level $\sigma_z^2$ approaches zero, the term $\frac{\min(k,l)}{2k} \log\frac{\sigma_a^2}{\sigma_z^2}$ exerts a major influence on the sufficient condition (\ref{suffCor1}). The nonzero signal matrix $W$ plays its role mainly through the ranks of its row-wise submatrices, which are ensured to be full rank according to the technical assumption that $\rank(\underline{W}_{\Tc}^\intercal \underline{W}_{\Tc})=\min(|\Tc|, l)$ for any $\Tc\subseteq[k]$.

Second, by rearranging the terms in (\ref{suffCor1}), we obtain
\[m = \left(\frac{\sigma_a^2}{\sigma_z^2}\right)^{ \alpha \cdot \min(k,l) \cdot \frac{n}{2k}}\]
which corresponds to the maximum number of columns of $A$ that still yields a diminishing error probability in support recovery. Specifically, the term $\min(k,l)$ reveals the following insight. In the scenario with sufficiently small $\sigma_z^2$, for the challenging problem where the number of measurement vectors is less than the number of nonzero rows, i.e., $l < k$, adding one more measurement vector can lead to a much larger upper bound on the manageable number of columns of $A$. On the other hand, when $k\leq l$, the problem is much simpler and adding more measurement vectors may not significantly increase the manageable size of $A$. From an algorithmic point of view, subspace based methods can be used to recover the support in the latter case.

\medskip
\subsubsection{The Role of the Nonzero Signal Matrix} Next, we take a closer look at on the role of the nonzero signal matrix $W$ in support recovery with MMV. We consider two different cases. In the first case, $W$ consists of identical columns. The following corollary states the corresponding sufficient condition for support recovery.
\begin{corollary}
Suppose $W\in\mathbb{R}^{k\times l}$ has identical columns, i.e., $W = [\wv,...,\wv]$, for some $\wv\in\mathbb{R}^k$ with all entries being nonzero. If
\begin{align}
\lim_{ m\rightarrow \infty}\frac{\log m}{n_m} &< \min_{\Tc\subseteq[k]}\frac{1}{2|\Tc|} \log \left(1+l\cdot\frac{\sigma_a^2}{\sigma_z^2}\|\wv_{\Tc}\|_2^2\right)\label{suffCor2}
\end{align}
then there exists a sequence of support recovery maps $\{d^{(m)}\}_{m=k}^\infty,d^{(m)}:\mathbb{R}^{n_m\times l}\mapsto 2^{[m]}$, such that
\begin{align}
\lim_{m\rightarrow \infty}\P\{d(Y)\neq \textmd{supp}(X(W, \Sv))\} =0.\nonumber
\end{align}
\end{corollary}
\begin{proof}
Note that, for any $\Tc\subseteq[k]$,
\begin{align}
\log\det\left(I + \frac{\sigma_a^2}{\sigma_z^2}\underline{W}_{\Tc}^\intercal \underline{W}_{\Tc}\right)&= \log\det\left(I + \frac{\sigma_a^2}{\sigma_z^2}[\wv_{\Tc}, ...,\wv_{\Tc}]^\intercal [\wv_{\Tc}, ..., \wv_{\Tc}]\right)\nonumber\\
&= \log \det \left(I+\frac{\sigma_a^2}{\sigma_z^2}\|\wv_{\Tc}\|_2^2 \mathbf{1}\cdot\mathbf{1}^\intercal   \right)\nonumber\\
&= \log\left(1 + l\cdot\frac{\sigma_a^2}{\sigma_z^2}\|\wv_{\Tc}\|_2^2\right)\nonumber.
\end{align}
Applying Theorem 1 completes the proof.
\end{proof}

Based on (\ref{suffCor2}), the effect of having $l$ identical nonzero signal vectors is equivalent to decreasing the noise level by a factor of $l$, compared to the problem with SMV. This is in accordance with the intuition that when the underlying signals remain the same, taking more measurement vectors provides an opportunity to average down the measurement noise level. We hasten to add that identical columns are unlikely in practice. Even small changes in the coefficients can lead to a full rank matrix, leading to significant benefits in the high signal-to-noise ratio (SNR) case.

In the second case, we construct a special example to achieve a large performance improvement via a second measurement. This is demonstrated in the following corollary.
\begin{corollary}
\label{corollary3}
Suppose $W = [\wv_1, \wv_2]\in\mathbb{R}^{k\times 2}$, where $k$ is even, $\wv_1 = \mathbf{1}\in\mathbb{R}^k$, and $\wv_2$ is defined as
\begin{align}
w_{i, 2} = \left\{ \begin{array}{ll}
1 & \mbox{if $1<i\leq \frac{k}{2}$},\\
-1 & \mbox{if $\frac{k}{2}<i\leq k $}.\end{array} \right.
\label{definew2}
\end{align}
If
\begin{align}
\lim_{m\rightarrow \infty} \frac{\log m}{n} < \frac{1}{k} \log\left(1+k\cdot\frac{\sigma_a^2}{\sigma_z^2}\right)
\end{align}
then there exists a sequence of support recovery maps $\{d^{(m)}\}_{m=k}^\infty,d^{(m)}:\mathbb{R}^{n_m\times l}\mapsto 2^{[m]}$, such that
\begin{align}
\lim_{m\rightarrow \infty}\P\{d(Y)\neq \textmd{supp}(X(W, \Sv))\} =0.\nonumber
\end{align}
\end{corollary}

\begin{proof}
Please see Appendix \ref{A_ProofofExample}.
\end{proof}

For the ease of illustration, we compare the performances among the problems with (i) SMV where $W = \mathbf{1}\in \mathbb{R}^{k\times 1}$, (ii) MMV where $W = [\mathbf{1}, \mathbf{1}]\in\mathbb{R}^{k\times 2}$, and (iii) MMV where $W$ is defined in Corollary 3, for an even $k$.\footnote{Note that $\|\wv_1\|_2 = \|\wv_2\|_2$, which can be viewed as a way of normalization to make comparison meaningful.} The following table summarizes the results.
\begin{table}[h]
\renewcommand{\arraystretch}{1.5}
\centering
\begin{tabular}{l||c|c}
\hline
 & lower bound on $n$ & upper bound on $m$\\
\hline
\hline
(i) SMV ($W = \mathbf{1}$)& $n > \frac{\log m}{\frac{1}{2k}\log\left(1 + \frac{k\sigma_a^2}{\sigma_z^2}\right)}$ & $m < \left(1 + \frac{k\sigma_a^2}{\sigma_z^2}\right)^\frac{n}{2k}$\\
\hline
(ii) MMV ($W = [\mathbf{1}, \mathbf{1}]$, Corollary 2) & ${n} > \frac{\log m}{\frac{1}{2k}\log\left(1 + 2\cdot\frac{k\sigma_a^2}{\sigma_z^2}\right)}$ & $m < \left(1 + 2\cdot\frac{k\sigma_a^2}{\sigma_z^2}\right)^\frac{n}{2k}$\\
\hline
(iii) MMV ($W$ as defined in Corollary 3) & ${n} > \frac{\log m}{\frac{1}{k}\log\left(1 + \frac{k\sigma_a^2}{\sigma_z^2}\right)}$ & $m < \left(1 + \frac{k\sigma_a^2}{\sigma_z^2}\right)^{\frac{n}{k}}$\\
\hline
\end{tabular}
\end{table}

Based on this table, we have the following observations for this specific setup. First, compared with the SMV problem, having MMV can improve the performance of support recovery by enabling a relaxed condition on the number of measurements $n$. Equivalently, for the same number of measurements per measurement vector, the MMV setup permits a measurement matrix $A$ with more columns. Second, the performance improvement enabled by having MMV is closely related to $c(W)$, and it can be quite different for different nonzero signal value matrices. In case (ii), we achieve a moderate performance gain which is equivalent to reducing the noise level by half. On the contrary, in case (iii), a larger performance gain can be achieved due to the structure of the nonzero signal value matrix. Note that the change occurs in the factor in the exponent in the upper bound for $m$. In summary, these examples are specially constructed as representative cases to illustrate the effect of the nonzero signal value matrix $W$ in support recovery. Generally, the difficulty of a support recovery problem is inherently determined by the model parameters and Theorems 1 and 2 together characterize their exact roles.

\medskip
\subsubsection{A Generalization of $W$}\label{sec:generalW}
Thus far, we have assumed $w_{i,j}\neq 0$ for all $i, j$ in the discussion above. Now, we generalize $W$ in the following manner: for each $i\in[k]$, there exist a $j\in[l]$ such that $w_{i,j}\neq 0$; meanwhile, for each $j\in[l]$, there exist a $i\in[k]$ such that $w_{i,j}\neq 0$. This relaxed assumption indicates that neither a zero row nor a zero column exists but zero elements are allowed in $W$, as opposed to the original assumption that all elements of $W$ are nonzeros. Accordingly,
\begin{align}
\supp(X) = \bigcup_{j=1}^l\supp(\Xv_j)\nonumber
\end{align}
which means the support of $X$ is equivalent to the union of the supports of all columns of $X$. Following the proofs for Theorem 1 and 2, one can readily see that the two theorems still hold in this case.

It is worthwhile to note that having more measurement vectors does not necessarily result in performance improvement. To illustrate this point, we construct a simple example. Let $W^{(1)} =  \left[ 0.1, 5\right]^\intercal$, $W^{(2)} =  \bigg[ \begin{array}{cc}
0.1& 0 \\
5& 6 \end{array} \bigg]$, and $\frac{\sigma_a^2}{\sigma_z^2} = 10$. As a result, $c(W^{(1)}) = c(W^{(2)}) = \frac{1}{2}\log 1.1$. This means that the performance limits for these two setups are the same. Intuitively, by inspecting the definition of $c(W)$, it can be seen that if a submatrix composed of certain rows of $W$ is ill-conditioned, the minimization inside $c(W)$ may likely be determined by that submatrix. Hence, for an extra measurement vector to benefit support recovery, this measurement vector should correspond to a column of $W$ whose presence improves the small eigenvalues of the previous worst-case submatrix that causes the performance bottleneck. The observations are reminiscent of some of the intuition developed in space time wireless communication systems \cite{Paulraj}.
The $l$ receivers can be viewed an a $l$ antenna receiver and it is known that the rank of the channel matrix plays an important role
in the high SNR case. The correlation between the channel gains is not as harmful in this context. The gains of having multiple receive antennas is lower at low SNR \cite{Paulraj}.

\section{Relation to Existing Results}
\label{sec:relation}
We discuss the relation between the main results in this paper and existing results in the literature.

\subsection{Relation to the Performance of Practical Algorithms}
Our analysis provides the performance limit that governs all possible support recovery algorithms. This is achieved by a theoretic support recovery method which has exponential complexity and therefore is impractical. However, it is interesting to make comparisons with performance limits of practical algorithms, since it provides insight into the potential gap between the performance of a practical algorithm and the fundamental performance limit, and suggests possibilities for performance improvement.

We note that the model employed in Obozinski, Wainwright, and Jordan \cite{OWainJ} is similar to the measurement model (\ref{CS_model}). Sufficient and necessary conditions are derived therein for multivariate group Lasso to successfully recover the support of the sparse signal in the presence of noise, as $m$, $n$, and $k$ jointly grow to infinity in certain manner.\footnote{Note that it is stated at the end of Section 3.3 of \cite{OWainJ} that the requirement on $k$ growing to infinity can be removed. The remark therein provided an alternative probability upper bound for the intermediate term $T_1$ such that this bound can drop to zero even for a fixed $k$. However, it seems that the other intermediate term $T_2$ still relies on a probability upper bound that involves a term scaling as $\exp(-\frac{k}{2})$, which requires an increasing $k$ to drive it to zero.} This is different from our assumption that $k$ is fixed. Although a direct comparison may seem difficult, we wish to draw the following intuitive discussion. Note that Example 1 in \cite[Section 2.3]{OWainJ} considered the case for identical regression, which means the nonzero signal matrix $W$ has identical columns. The conclusion therein is that multivariate group Lasso offers no performance improvement under the MMV formulation compared with using Lasso on an SMV formulation with one measurement vector. However, our Corollary 2 indicates that the effect of having $l$ identical columns in $W$ is equivalent to lowering the noise level by a factor of $l$. 
The different performances indicated by multivariate group Lasso and the information theoretic analysis lead to the following observation. In general, if the sparse signal to be recovered possesses strong structural property, an algorithm needs to take advantage of this factor in order to achieve better performance. For multivariate group Lasso, the $\ell_1/\ell_p$ cost term completely ignores the row-wise structure presented in the nonzero entries. In contrast, AR-SBL \cite{ZhilinRao} is developed based on the assumption that the elements of $W$ are drawn from an auto-regressive process, and it explicitly attempts to learn this correlation structure. Based on the experimental study presented in \cite{ZhilinRao}, notable performance improvement in support recovery was observed when such correlation is present, including the case when the columns of $W$ were highly correlated.

\medskip
\subsection{Relation to Information Theoretic Performance Analysis}
\label{sec:relationInfoTheo}

Under the assumption that $\sigma_a^2 = 1$ and the elements of $W$ are i.i.d. according to $\mathcal{N}(0,1)$,\footnote{We only consider the real case in this discussion.} Tang and Nehorai \cite{TangNehorai} identifies sufficient and necessary conditions, involving the model parameters (i.e., $m, n, k, l,$ and $ \sigma_z^2$), to ensure diminishing error probability in support recovery as the problem size grows to infinity. We restate the sufficient condition to facilitate the discussion.

\begin{theorem}[{\cite[Theorem 5]{TangNehorai}}]
\label{theorem3}
Suppose that $n = \Omega(k\log\frac{m}{k})$ and $\frac{l}{2}\log\frac{n}{\sigma_z^2}\gg \log(k(m-k))$, then with probability one the error probability vanishes. In particular, if $n = \Omega(k\log\frac{m}{k})$ and $l\gg\frac{\log m}{\log\log m}$, the error probability vanishes as $m\rightarrow \infty$.
\end{theorem}

\medskip

As noted in \cite{TangNehorai}, heuristically, when $l = 1$, $n\gg m$ is needed to guarantee asymptotically successful support recovery. Although our main results aim for the case with fixed $W$, intuitive observations can still be drawn to provide more insight into the behavior of the support recovery with random $W$. To see this, recall that for a sequence of support recovery problems with a fixed $W$, the quantity $c(W)$ inherently determines the performance limit and the sufficient condition is $n > \frac{1}{c(W)}\log m$. Now, let us assume that the elements of $W$ are i.i.d. according to certain distribution with bounded support. Thus, in general, for any constant $\delta>0$, the probability $\P(c(W)<\delta)$ may be strictly positive. This implies that for the scaling $n = \Theta(\log m)$, the error probability will not converge to zero because there is a nontrivial probability of poor realizations of $W$ such that the sufficient condition above cannot be satisfied. As one plausible solution, we need $n$ to grow with $m$ at a much faster rate to ensure that the sufficient condition above can be met with probability converging to one.

\section{Conclusion and Discussion}
\label{sec:discussion}
We have developed performance limits for support recovery of sparse signals when multiple measurement vectors are available. Sufficient and necessary conditions are obtained for support recovery to be asymptotically successful. Especially, the role of nonzero entries in the performance limits is explicitly characterized, and the quantity $c(W)$ captures the effect of all nonzero entries. The key technique that enabled our analysis is motivated by the connection between sparse signal recovery with MMV and multiple access communication over SIMO channels. This leads to the opportunity of leveraging the methodology for deriving SIMO MAC capacity to help understand the performance limits of sparse signal recovery with MMV. Interpretations of the main results were provided in order to demonstrate the performance improvement by having MMV, and relations to existing results were also discussed.

The proposed methodology also has the potential to address other theoretical and practical issues associated with sparse signal recovery. First, this analytical approach can be extended to deal with the case where the signal value matrix $W$ is random. Outage analysis for fading channels can be leveraged to reveal the performance limits for sparse signal recovery in this case. Second, one can consider the problem where recovering a partial support is also desirable, if recovering the full support is not possible \cite{Reeves_08ISIT}. This can be achieved by treating a subset of users as noise and examining the capacity region of the remaining users. The connection between sparse signal recovery and multiple access communication offers the opportunity to explore the connection between sparse recovery algorithms and multiuser detection techniques with potential for cross-fertilization. A sender with larger channel gain may be easier to detect compared to a sender with weaker channel gain. The successive interference cancellation (SIC) scheme is aimed to detected users in a sequential manner, where the remaining undetected users are treated as noise bearing a strong resemblance to the matching pursuit algorithms for sparse signal recovery. It is conceivable that by appropriately utilizing the techniques for channel coding, performance limits could be obtained for partial support recovery of sparse signals.

Further, according to the interpretations of the main results, we can see that the structure of $W$ plays an important role in the performance limits. Roughly speaking, high correlation among the columns of $W$ may decrease the performance limit for support recovery, in the sense that, given other parameters fixed, the dimension of the signal $m$ should be reduced to guarantee successful support recovery. However, as observed in practice, when only a finite number of measurements per measurement vector are available, a strong correlation among columns of $W$ actually facilitates the estimation of the nonzero signal values, and hence can be beneficial to the performance. Hence, there is an interplay that is not revealed by the asymptotic analysis. It will be interesting to study an analytical approach which links the estimation quality of nonzero values in the finite case and performance limits of support recovery in the asymptotic case.

\section*{Acknowledgment}
This research was supported by NSF Grants CCF-0830612 and CCF-1144258. The authors wish to thank Professor Young-Han Kim for insightful discussions on the problem of sparse signal recovery and its connection to multiple access communication. The proof techniques extend upon our previous methodology for the SMV case, to which Professor Kim made significant contributions.

\newpage
\appendices
\section{Proof of Theorem 1}
\label{app:th1}
For the ease of exposition, we consider two distinct cases on the number of nonzero rows of $X$.

\textbf{Case 1}: $k=1$. In this case, the signal of interest is $X = X(W, S_1)$, where $W = [w_{1,1},...,w_{1, l}]$. Fix $\epsilon>0$. We first form an estimate $\hat{\rho}_i$ of $|w_{1,i}|$ for $i\in[l]$ as
\begin{align}
\hat{\rho}_i \triangleq \sqrt{\frac{|\frac{1}{n_m}\|\Yv_i\|_2^2-\sigma_z^2|}{\sigma_a^2}}.
\end{align}

Declare that $\hat{s}_1\in[m]$ is the estimated index of the nonzero row, i.e., $d^{(m)}(Y)=\{\hat{s}_1\}$, if it is the unique index such that
\begin{align}
\frac{1}{nl}\left\|Y - \Av_{\hat{s}_1}\left[(-1)^{q_1}\hat{\rho}_1, ..., (-1)^{q_l}\hat{\rho}_l\right]\right\|_F^2\leq \sigma_z^2 + \epsilon^2\sigma_a^2
\end{align}
for $q_i = 1$ or $q_i = 2$, $i\in [l]$. If there is none or more than one such index, pick an arbitrary index.

We analyze the average probability of error
\begin{align}
\P(\Ec) = \P\{d^{(m)}(Y)\neq \textmd{supp}(X(W, S_1))\}.
\end{align}
Due to the symmetry in the problem and the measurement matrix generation, we assume without loss of generality $S_1 = 1$, that is,
\begin{align}
Y = \Av_1W + Z
\end{align}
for some $W = [w_{1, 1}, ...,w_{1, l}] \in\mathbb{R}^{1\times l}$.
In the following analysis, we drop superscripts and
subscripts on $m$ for notational simplicity when no ambiguity arises. Define the events
\begin{align}
\Ec_{s} \triangleq \left\{\forall i\in[l],\exists q_i\in\{1,2\}, \text{ such that }
\frac{1}{nl}\left\|Y - \Av_s[(-1)^{q_i}\hat{\rho}_1, ..., (-1)^{q_l}\hat{\rho}_l]  \right\|_F^2 \leq \sigma_z^2
+ \epsilon^2  \sigma^2_a\right\}, ~s\in[m].\nonumber
\end{align}
Then,
\begin{align}
\P(\Ec) \le \P\left(\Ec_{1}^c \cup\left(\cup_{s=2}^{m} \Ec_{s} \right)\right)
\label{error_events}
\end{align}
where $\Ec^c$ denotes the compliment event of $\Ec$.
Let
\begin{align}
\Ec_{\text{aux}}&\triangleq \left\{    \textmd{det}\left(\frac{1}{n}\left(\Av_1 W + Z\right)^\intercal \left(\Av_1 W + Z\right)\right)  - \textmd{det}\left(\sigma_a^2 W^\intercal W + \sigma_z^2 I\right)  \in (-\epsilon, \epsilon)  \right\}\nonumber\\
& ~~~~~\cap
\left(\bigcap_{i=1}^l \left\{\hat{\rho}_i-|w_{1,i}|\in\left(-\epsilon,\epsilon\right)\right\}\right).\nonumber
\end{align}
Then, by the union of events bound and the fact that $\Ac^c\cup\Bc =
\Ac^c\cup (\Bc \cap \Ac)$,
\begin{align}
\P(\Ec)
&\leq \P(\Ec_\text{aux}^c) + \P(\Ec_1^c) + \sum_{s=2}^m \P(\Ec_s \cap \Ec_\text{aux}).
\label{C1_Error}
\end{align}

We bound each term in (\ref{C1_Error}). First, by the weak law of large numbers (LLN), $\lim_{m\rightarrow\infty}\P(\Ec_\text{aux}^c) = 0$. Next, we consider $\P(\Ec_1^c)$. It can be readily seen that, with $q_i = (3 + \textrm{sign}(w_{1,i}))/2$,
\begin{align}
\lim_{m\rightarrow\infty}\P\left(\frac{1}{nl}\left\|Y - \Av_1[(-1)^{q_1}\hat{\rho}_1, ..., (-1)^{q_l}\hat{\rho}_l]  \right\|_F^2 \leq \sigma_z^2
+ \epsilon^2  \sigma^2_a \right)  = 1.
\end{align}
Hence, $\lim_{m\rightarrow\infty}\P(\Ec_1^c) =0$.

Next, we consider the third term in (\ref{C1_Error}). We need the following lemma, whose proof is presented at the end of this appendix.
\begin{lemma}
\label{lemma1}
Let $B\in\mathbb{R}^{n\times l}$ be a fixed matrix satisfying $(\prod_{i=1}^l [\frac{1}{n}B^\intercal B]_{i,i})^\frac{1}{l} \equiv \alpha>0$. Let $\Sc\subseteq[l]$ be a fixed set. Let $D\in\mathbb{R}^{n\times l}$ be a matrix such that, for $j\in\Sc$, $\Dv_j\sim\mathcal{N}(\mathbf{0}, \theta_j I)$ with some $\theta_j>0$; for $j\in[l]\backslash\Sc$, $\Dv_j \equiv \mathbf{0}$. All columns of $D$ are independent. Then, for any $\gamma\in(0, \alpha)$,
\begin{align}
\P\left(\frac{1}{nl}\|B-D\|_F^2 \leq \gamma \right)\leq 2^{-\frac{n}{2}\log\frac{\alpha^l}{\gamma^l}}.
\end{align}

\end{lemma}

We continue the proof of Theorem 1. Consider $\P(\Ec_s\cap\Ec_\text{aux})$ for $s \ne 1$. Note that
\begin{align*}
\P(\Ec_s\cap\Ec_\text{aux}) \leq \P(\Ec_s|\Ec_\text{aux}) =\int_{Y_1\in \Ec_\text{aux}} \P(\Ec_s|\{Y =Y_1\} \cap \Ec_\text{aux})f(Y_1|\Ec_\text{aux})dY_1.\nonumber
\end{align*}
Let $[(-1)^{q_1} \hat{\rho}_1, ...,(-1)^{q_l} \hat{\rho}_l] = U\Theta V^\intercal$ denote the singular value decomposition. Since $\Av_s$ is independent of $Y$ and $\hat{\rho}_i$ for $s \neq 1$, it follows from Lemma \ref{lemma1} that (by treating $B = YV$ and $D = \Av_s U\Theta$), for $q_i \in \{1, 2\},i\in[l]$ and sufficiently small $\epsilon$,
\begin{align}
&\P\left(\frac{1}{nl}\left\|Y - \Av_s [(-1)^{q_1} \hat{W}_1, ...,(-1)^{q_l} \hat{W}_l]\right\|_F^2\leq \sigma_z^2 + \epsilon^2 \sigma^2_a
\Big| \{Y =Y_1\} \cap\Ec_\text{aux}\right)\nonumber\\
&= \P\left(\frac{1}{nl}\left\|YV - \Av_s U\Theta\right\|_F^2\leq \sigma_z^2 + \epsilon^2 \sigma^2_a
\Big| \{Y =Y_1\} \cap\Ec_\text{aux}\right)\nonumber\\
&\leq 2^{-\frac{n}{2}\log\frac{\prod_{i=1}^l [\frac{1}{n}V^\intercal Y^\intercal Y V]_{i,i}}{(\sigma_z^2 + \epsilon^2 \sigma^2_a)^l}}\nonumber\\
&\leq 2^{-\frac{n}{2}\log\frac{\det(\frac{1}{n}V^\intercal Y^\intercal Y V)}{(\sigma_z^2 + \epsilon^2 \sigma^2_a)^l}}\label{explain2}\\
&\leq 2^{-\frac{n}{2}\log\frac{\det(\frac{1}{n} Y^\intercal Y )}{(\sigma_z^2 + \epsilon^2 \sigma^2_a)^l}}\nonumber\\
&\leq  2^{-\frac{n}{2}\log\left(\frac{\textmd{det}\left(\sigma_a^2 W^\intercal W + \sigma_z^2 I\right)-\epsilon}{(\sigma_z^2 + \epsilon^2 \sigma^2_a)^l} \right)}\nonumber
\end{align}
where (\ref{explain2}) follows from the Hadamard's inequality \cite{Cover06}. Thus,
\[
\P(\Ec_s|\{Y =Y_1\} \cap\Ec_\text{aux})\leq 2^l\cdot 2^{-\frac{n}{2}\log\left(\frac{\textmd{det}\left(\sigma_a^2 W^\intercal W + \sigma_z^2 I\right)-\epsilon}{(\sigma_z^2 + \epsilon^2 \sigma^2_a)^l} \right)}\nonumber
\]
and hence
\[
\sum_{s=2}^m \P(\Ec_s \cap \Ec_\text{aux})
\leq 2^l \cdot m \cdot 2^{-\frac{n}{2}\log\left(\frac{\textmd{det}\left(\sigma_a^2 W^\intercal W + \sigma_z^2 I\right)-\epsilon}{(\sigma_z^2 + \epsilon^2 \sigma^2_a)^l} \right)}\nonumber
\]
which tends to zero as $m \to \infty$, if
\begin{equation}
\limsup _{m\rightarrow \infty } \frac{\log m}{n_m} < \frac{1}{2}\log\left(\frac{\textmd{det}\left(\sigma_a^2 W^\intercal W + \sigma_z^2 I\right)-\epsilon}{(\sigma_z^2 + \epsilon^2 \sigma^2_a)^l} \right).
\label{case1_final}
\end{equation}
Since $\epsilon > 0$ is
chosen arbitrarily, we have the desired proof of Theorem~1.

\medskip

\textbf{Case 2}: $k\geq 2$. In this case, the signal of interest is $X = X(W, \Sv)$. Fix $\epsilon>0$. First, for $i\in[l]$, we form an estimate of $\|\wv_i\|_2$ as
\begin{align}
\hat{\rho}_i \triangleq \sqrt{\frac{|\frac{1}{n}\|\Yv_i\|_2^2 - \sigma_z^2|}{\sigma_a^2}}.
\end{align}

For $r,\zeta >0$, let $\Qc = \Qc(r, \zeta)$ be a minimal set of points in $\mathbb{R}^k$ satisfying the following properties:
\begin{enumerate}
\item[i)] $\Qc \subseteq \mathcal{B}_k(r)$, where $\mathcal{B}_k(r)$ is the $k$-dimensional hypersphere of radius $r$.
\item[ii)] For any $\bv \in \Bc_k(r)$, there exists $\hat{\wv}\in\mathcal{Q}$ such that $\|\hat{\wv}-\bv\|_2\leq \frac{\zeta}{2}$.
\end{enumerate}

The following properties can be easily proved:
\begin{lemma}
\label{lemma2}
1) For $i\in[l]$, $\lim_{m\rightarrow\infty}\P\left(\exists \hat{\Wv}\in\Qc(\hat{W}_i,\zeta)\textrm{ such that }\|\hat{\Wv}-\wv_i\|_2 < \zeta\right) = 1.$

2) $q(r,\zeta)\triangleq|\Qc(r,\zeta)|$ is monotonically non-decreasing in $r$ for fixed $\zeta$.
\end{lemma}

For $i\in[l]$, given $\hat{\rho}_i$ and $\epsilon$, fix $\Qc_i = \Qc_i(\hat{\rho}_i, \epsilon)$. Declare $d(Y)=\{\hat{s}_1, ...,\hat{s}_k\}\subseteq[m]$ is the recovered set of indices of nonzero rows of $W$, if it is the unique set of indices such that
\begin{align}
\frac{1}{nl}\left\|Y - [\Av_{\hat{s}_1}, ...,\Av_{\hat{s}_k}]\left[\hat{\Wv}_1, ..., \hat{\Wv}_l\right]\right\|_F^2\leq \sigma_z^2 + \epsilon^2\sigma_a^2
\end{align}
for some $\hat{\Wv}_i\in\Qc_i$, $i\in[l]$. If there is none or more than one such set, pick an arbitrary set of $k$ indices.

Next, we analyze the average probability of error
\begin{align}
\P(\Ec) = \P\{d(Y) \neq X(W, \Sv)\}.
\end{align}
Without loss of generality, we assume that $S_j = j$ for $j = 1, 2, ..., k$, which gives
\begin{align}
Y = [\Av_1,...,\Av_k ]W + Z
\end{align}
for some $W$. Define the event
\begin{align}
&\Ec_{s_1, s_2,...,s_k}\triangleq \nonumber\\
& \left\{\exists \hat{\Wv}_i\in \Qc_i \textrm{ and } \{s'_1,...,s'_k\}= \{s_1, ...,s_k\}\textrm{ s.t. }\frac{1}{nl}\left\|Y - [\Av_{{s}_1'}, ...,\Av_{{s}_k'}]\left[\hat{\Wv}_1, ..., \hat{\Wv}_l\right] \right\|_F^2\leq \sigma_z^2+\epsilon^2\sigma_a^2\right\}.\nonumber
\end{align}

Define $\sigma_\textrm{max}^2$ and $\sigma_\textrm{min}^2$ to be the largest and smallest eigenvalues of the matrix
\[\frac{1}{{n\sigma_a^2}}[\Av_1,...,\Av_k, \frac{\sigma_a}{\sigma_z}Z]^\intercal[\Av_1,...,\Av_k,\frac{\sigma_a}{\sigma_z}Z]\]
respectively. Then
\begin{align}
\P(\Ec) &=  \P\left(\Ec_{1, 2,...,k}^c \cup \left(\bigcup_{s_1<\cdots<s_k:\{s_1, ..., s_k\}\neq[k]} \Ec_{s_1, s_2,...,s_k} \right)\right)\nonumber\\
&\leq  \P\left({\Ec^c_\text{aux}}\cup \Ec_{1, 2,...,k}^c   \cup\left(\bigcup_{s_1<\cdots<s_k:\{s_1, ..., s_k\}\neq[k]} (\Ec_{s_1, s_2,...,s_k}\cap\Ec_\text{aux}) \right)\right)\nonumber\\
&\leq  \P({\Ec^c_\text{aux}}) +
\P(\Ec_{1, 2,...,k}^c)+\sum_{s_1<\cdots<s_k:\{s_1, ..., s_k\}\neq[k]} \P(\Ec_{s_1, s_2,...,s_k}\cap\Ec_\text{aux})
\label{C2new_Perror}
\end{align}
where
\[
\Ec_{\text{aux}}\triangleq
\left\{ \sigma_{\textrm{max}}^2 \in \left(1 - \epsilon, 1 + \epsilon\right) \right\} \cap \left\{ \sigma_{\textrm{min}}^2 \in \left(1 - \epsilon, 1 + \epsilon\right) \right\}\cap \left( \bigcap_{i=1}^l\left\{\hat{W}_i-\|\wv_i\|_2\in\left(-\epsilon,\epsilon\right)\right\}\right).\nonumber
\]

First, note that $\lim_{m\rightarrow\infty}\P(\Ec_{\text{aux}}) = 1$ due to LLN and the properties of the extreme eigenvalues of random matrices \cite{Silverstein}. Next, consider
\begin{align}
& \frac{1}{nl}\left\|Y-[\Av_1, ..., \Av_k]\left[\hat{\Wv}_1,..., \hat{\Wv}_l\right]  \right\|_F^2\nonumber\\
&= \frac{1}{nl}\left\|[\Av_1, ..., \Av_k]W+Z-[\Av_1, ..., \Av_k]\left[\hat{\Wv}_1,..., \hat{\Wv}_l\right]  \right\|_F^2\nonumber\\
&= \frac{1}{nl}\left\|[\Av_1, ..., \Av_k, \frac{\sigma_a}{\sigma_z}Z]  \left[\begin{array}{c} W-\left[\hat{\Wv}_1,..., \hat{\Wv}_l\right]\\\frac{\sigma_z}{\sigma_a}I_{l\times l}\end{array}\right]  \right\|_F^2\nonumber\\
&\leq \frac{1}{l}\sigma_\textrm{max}^2\sigma_a^2\left\|\left[\begin{array}{c} W-\left[\hat{\Wv}_1,..., \hat{\Wv}_l\right]\\\frac{\sigma_z}{\sigma_a}I_{l\times l}\end{array}\right]\right\|_F^2\nonumber\\
&= \sigma_\textrm{max}^2\left( \frac{\sigma_a^2}{l}\left\|W-\left[\hat{\Wv}_1,..., \hat{\Wv}_l\right]\right\|_F^2 + \sigma_z^2   \right)
\end{align}
By using the fact that $\sigma_\textrm{max}^2\rightarrow 1$ almost surely as $n\rightarrow \infty$ \cite{Silverstein} and Lemma 2-1), we have $\lim_{m\rightarrow\infty}\P(\Ec_{1, 2,...,k}^c)=0$.

Next, we consider $\P(\Ec_{s_1, s_2,...,s_k}\cap {\Ec_{\text{aux}}})$ for $\{s_1, s_2,...,s_k\} \neq [k]$. Note that
\begin{align}
&\P(\Ec_{s_1, s_2,...,s_k}\cap \Ec_{\text{aux}})\nonumber\\
&\leq \P(\Ec_{s_1, s_2,...,s_k}| \Ec_{\text{aux}})\nonumber\\
&= \int\cdots\int_{\{\av_1, ..., \av_k, Z_0\}\in \Ec_{\text{aux}}}  \P(\Ec_{s_1, s_2,...,s_k} | \{\Av_1=\av_1\}\cap\cdots\cap\{\Av_k=\av_k\}\cap\{Z=Z_0\}\cap\Ec_{\text{aux}})\nonumber\\
&~~~~~~~~~~~~~~~~~~~~~~~~~~~~~\times f(\av_1,...,\av_k,Z_0|\Ec_{\text{aux}})d\av_1\cdots d\av_kdZ_0.
\label{C3_oneside_prob}
\end{align}
For notational simplicity, define $\xi\triangleq\sigma_z^2+\epsilon^2 \sigma_a^2$, ${\Tc} \triangleq \{s_1, s_2,...,s_k\} \cap [k]$, ${\Tc}^c \triangleq \{s_1, s_2,...,s_k\}\backslash {\Tc}$, and $\Ec_{\textrm{cond}} \triangleq \{\Av_1=\av_1\}\cap\cdots\cap\{\Av_k=\av_k\}\cap\{Z=Z_0\}\cap\Ec_{\text{aux}}$. For any permutation $(s_1',s_2', ..., s_k')$ of $\{s_1,s_2, ..., s_k\}$ and any $\hat{\Wv}_i\in\mathcal{Q}_i$, $i\in[l]$,
\begin{align}
& \P\left( \frac{1}{nl}\left\| Y- [\Av_{s_1'}, ..., \Av_{s_k'}]\left[\hat{\Wv}_1,..., \hat{\Wv}_l\right]   \right\|_F^2\leq \xi\Big| \Ec_\textrm{cond}   \right)\nonumber\\
& = \P\left( \frac{1}{nl}\left\|  [\Av_1, ..., \Av_k]W+Z - [\Av_{s_1'}, ..., \Av_{s_k'}]\left[\hat{\Wv}_1,..., \hat{\Wv}_l\right]   \right\|_F^2  \leq \xi\Big| \Ec_\textrm{cond}   \right)\label{temp1}
\end{align}
Define the matrix $W'\in\mathbb{R}^{k\times l}$ as
\begin{align}
\underline{W}'_j = \left\{ \begin{array}{ll}
\underline{W}_j & \mbox{if $j\in[k]\backslash \Tc$}\\
\underline{W}_j - \hat{\underline{\Wv}}_i & \mbox{if $j=s_i' \in \Tc$}\end{array} \right.
\end{align}
where $\hat{\underline{\Wv}}_i$ denotes the $i$th row of the matrix $\left[\hat{\Wv}_1,..., \hat{\Wv}_l\right]$. Define $\widetilde{W}'\in\mathbb{R}^{k\times l}$ as
\begin{align}
\underline{\widetilde{W}}'_j = \left\{ \begin{array}{ll}
\hat{\underline{\Wv}}_j & \mbox{if $s_j'\notin\Tc$}\\
\underline{\mathbf{0}} & \mbox{if $s_j' \in \Tc$}\end{array} \right.
\end{align}
where $\underline{\mathbf{0}}$ is a zero row vector of a proper size. Then, continue from (\ref{temp1}), we have
\begin{align}
& \P\left( \frac{1}{nl}\left\| Y- [\Av_{s_1'}, ..., \Av_{s_k'}]\left[\hat{\Wv}_1,..., \hat{\Wv}_l\right]   \right\|_F^2\leq \xi\Big| \Ec_\textrm{cond}   \right)\nonumber\\
& = \P\left( \frac{1}{nl}\left\|  [\Av_1, ..., \Av_k, \frac{\sigma_a}{\sigma_z}Z] \left[\begin{array}{c} W'\\\frac{\sigma_z}{\sigma_a}I\end{array}\right] - [\Av_{s_1'}, ..., \Av_{s_k'}]\widetilde{W}'  \right\|_F^2  \leq \xi\Big| \Ec_\textrm{cond}   \right)\nonumber\\
& \equiv \P\left( \frac{1}{nl}\left\|  [\Av_1, ..., \Av_k, \frac{\sigma_a}{\sigma_z}Z] \left[\begin{array}{c} W'\\\frac{\sigma_z}{\sigma_a}I\end{array}\right] - \widetilde{A} \widetilde{W}_1'  \right\|_F^2  \leq \xi\Big| \Ec_\textrm{cond}   \right)\label{explain3}\\
&= \P\left( \frac{1}{nl}\left\|  [\Av_1, ..., \Av_k, \frac{\sigma_a}{\sigma_z}Z] \left[\begin{array}{c} W'\\\frac{\sigma_z}{\sigma_a}I\end{array}\right]V - \widetilde{A} U\Theta  \right\|_F^2  \leq \xi\Big| \Ec_\textrm{cond}   \right)\label{explain4}
\end{align}
where in (\ref{explain3}) $\widetilde{W}_1'$ denotes matrix formed by removing the zero rows in $\widetilde{W}'$, and $\widetilde{A}$ denotes the matrix by removing columns of $[\Av_{s_1'}, ..., \Av_{s_k'}]$ indexed by the indices of the zero rows of $\widetilde{W}'$. To reach (\ref{explain4}), let $\widetilde{W}_1' = U\Theta V^\intercal$ denote the singular value decomposition.
The follow lemma, the proof of which is presented at the end of this appendix, is useful.
\begin{lemma}
\label{lemma3}
Let $B\in\mathbb{R}^{p\times q}$, $D\in\mathbb{R}^{q\times r}$. Let $\sigma_b^2$ denote the smallest eigenvalue of $B^\intercal B$. Then
\[
\det((BD)^\intercal BD)\geq (\sigma_b^2)^r \det(D^\intercal D).
\]
\end{lemma}

Let $M \triangleq [\Av_1, ..., \Av_k, \frac{\sigma_a}{\sigma_z}Z] \left[\begin{array}{c} W'\\\frac{\sigma_z}{\sigma_a}I\end{array}\right]V$. Conditioned on $\Ec_\textrm{cond}$ and the chosen $\Qc_i$ for $i\in[l]$, $M$ is fixed. According to Lemma \ref{lemma3} (treating $B = \frac{1}{\sqrt{n}}[\Av_1, ..., \Av_k, \frac{\sigma_a}{\sigma_z}Z]$ and $D = \left[\begin{array}{c} W'\\\frac{\sigma_z}{\sigma_a}I\end{array}\right]V$),
\begin{align}
\det \left(\frac{1}{n} M^\intercal M \right) \geq ((1-\epsilon) \sigma_a^2)^l \det \left( \left[\begin{array}{c} W'\\\frac{\sigma_z}{\sigma_a}I\end{array}\right]^\intercal \left[\begin{array}{c} W'\\\frac{\sigma_z}{\sigma_a}I\end{array}\right]\right).
\end{align}
Continue with (\ref{explain4}). Using Lemma \ref{lemma2} (treating $B=M$ and $D=\widetilde{A} U\Theta $), we have
\begin{align}
&\P\left( \frac{1}{nl}\left\| Y- [\Av_{s_1'}, ..., \Av_{s_k'}]\left[\hat{\Wv}_1,..., \hat{\Wv}_l\right]   \right\|_F^2\leq \xi\Big| \Ec_\textrm{cond}   \right)\nonumber\\
&\leq 2^{-\frac{n}{2}\log\frac{\prod_{i=1}^l [\frac{1}{n} M^\intercal M]_{i,i}}{(\sigma_z^2+\epsilon^2 \sigma_a^2)^l} }\nonumber\\
&\leq 2^{-\frac{n}{2}\log\frac{\det (\frac{1}{n} M^\intercal M)}{(\sigma_z^2+\epsilon^2 \sigma_a^2)^l} }\nonumber\\
&\leq 2^{-\frac{n}{2}\log\frac{\left((1-\epsilon) \sigma_a^2\right)^l \det \left( \left[\begin{array}{c} W'\\\frac{\sigma_z}{\sigma_a}I\end{array}\right]^\intercal \left[\begin{array}{c} W'\\\frac{\sigma_z}{\sigma_a}I\end{array}\right]\right)}{(\sigma_z^2+\epsilon^2 \sigma_a^2)^l}}\nonumber\\
&\leq 2^{-\frac{n}{2}\log\frac{\left((1-\epsilon) \sigma_a^2\right)^l \det \left( \left[\begin{array}{c} \underline{W}_{[k]\backslash\Tc}\\\frac{\sigma_z}{\sigma_a}I\end{array}\right]^\intercal \left[\begin{array}{c} \underline{W}_{[k]\backslash\Tc}\\\frac{\sigma_z}{\sigma_a}I\end{array}\right]\right)}{(\sigma_z^2+\epsilon^2 \sigma_a^2)^l}}
\label{explain1}\\
&= 2^{-\frac{n}{2}\log\frac{\left((1-\epsilon) \sigma_a^2\right)^l \det \left(\underline{W}_{[k]\backslash\Tc}^\intercal \underline{W}_{[k]\backslash\Tc} + \frac{\sigma_z^2}{\sigma_a^2}I \right)}{(\sigma_z^2+\epsilon^2 \sigma_a^2)^l}}
\end{align}
where (\ref{explain1}) uses the fact that
\[
\left[\begin{array}{c} W'\\\frac{\sigma_z}{\sigma_a}I\end{array}\right]^\intercal \left[\begin{array}{c} W'\\\frac{\sigma_z}{\sigma_a}I\end{array}\right] = \left[\begin{array}{c} \underline{W}_{[k]\backslash\Tc}\\\frac{\sigma_z}{\sigma_a}I\end{array}\right]^\intercal \left[\begin{array}{c} \underline{W}_{[k]\backslash\Tc}\\\frac{\sigma_z}{\sigma_a}I\end{array}\right] + \left[\begin{array}{c} \underline{W}'_{\Tc}\\O\end{array}\right]^\intercal \left[\begin{array}{c} \underline{W}'_{\Tc}\\O\end{array}\right]
\]
where $O$ denotes the matrix with elements all being zeros, and the fact that \cite[Corollary 8.4.15]{matrixmath}, for positive semidefinite $B, D\in\mathbb{R}^{l\times l}$, $\det(B+D)\geq \det(B)$. By the union of events bound,
\begin{align}
&\P(\Ec_{s_1, s_2,...,s_k}| \Ec_{\textrm{cond}})\nonumber\\
&\leq  \sum_{{\{s_1',...,s_k'\}=\{s_1,...,s_k\}}}  \P\Bigg(\forall i, \exists \hat{\Wv}_i \in \Qc_i \textrm{ such that } \frac{1}{nl}\left\| Y- [\Av_{s_1'}, ..., \Av_{s_k'}]\left[\hat{\Wv}_1,..., \hat{\Wv}_l\right]   \right\|_F^2\leq \xi \bigg| \Ec_{\textrm{cond}}\Bigg)\nonumber\\
&\leq \sum_{{\{s_1',...,s_k'\}=\{s_1,...,s_k\}}} \sum_{\hat{\Wv}_1\in\Qc_1}\cdots\sum_{\hat{\Wv}_l\in\Qc_l} \P\Bigg(\frac{1}{nl}\left\| Y- [\Av_{s_1'}, ..., \Av_{s_k'}]\left[\hat{\Wv}_1,..., \hat{\Wv}_l\right]   \right\|_F^2\leq \xi \bigg|  \Ec_{\textrm{cond}}\Bigg)\nonumber\\
&\leq k!\cdot\left(\prod_{i=1}^l|\mathcal{Q}_i|\right)\cdot 2^{-\frac{n}{2}\log\frac{\left((1-\epsilon) \sigma_a^2\right)^l \det \left(\underline{W}_{[k]\backslash\Tc}^\intercal \underline{W}_{[k]\backslash\Tc} + \frac{\sigma_z^2}{\sigma_a^2}I \right)}{(\sigma_z^2+\epsilon^2 \sigma_a^2)^l}}.\nonumber
\end{align}
Furthermore, conditioned on $\Ec_{\text{aux}}$, $\hat{\rho}_i<\|\wv_i\|_2+\epsilon$ for $i\in[l]$ and hence $|\mathcal{Q}_i| \leq q_i(\|\wv_i\|_2+\epsilon, \epsilon)$ by Lemma \ref{lemma2}-2). Thus,
\begin{align}
\P(\Ec_{s_1, s_2,...,s_k}\cap \Ec_{\text{aux}})\leq k! \cdot \left(\prod_{i=1}^l q_i(\|\wv_i\|_2+\epsilon, \epsilon)\right) \cdot 2^{-\frac{n}{2}\log\frac{\left((1-\epsilon) \sigma_a^2\right)^l \det \left(\underline{W}_{[k]\backslash\Tc}^\intercal \underline{W}_{[k]\backslash\Tc} + \frac{\sigma_z^2}{\sigma_a^2}I \right)}{(\sigma_z^2+\epsilon^2 \sigma_a^2)^l}}.
\label{general_bound}
\end{align}
Note that the probability upper-bound (\ref{general_bound}) depends on $s_1,...,s_k$ only through $\Tc$. Grouping the ${m-k}\choose {k-|{\Tc}|}$ events $\{\Ec_{s_1, s_2,...,s_k}\cap \Ec_{\text{aux}}\}$ with the same $\Tc$,
\begin{align}
&\P(\Ec)\nonumber\\
&\leq  \P(\Ec^c_\text{aux}) +\P(\Ec_{1, 2, ..., k}^c) +\sum_{{\Tc}\subset[k]} {{m-k}\choose {k-|{\Tc}|}} \cdot k!\cdot \left(\prod_{i=1}^l q_i(\|\wv_i\|_2+\epsilon, \epsilon)\right) \cdot 2^{-\frac{n}{2}\log\frac{\left((1-\epsilon) \sigma_a^2\right)^l \det \left(\underline{W}_{[k]\backslash\Tc}^\intercal \underline{W}_{[k]\backslash\Tc} + \frac{\sigma_z^2}{\sigma_a^2}I \right)}{(\sigma_z^2+\epsilon^2 \sigma_a^2)^l}}\nonumber\\
&\leq  \P(\Ec^c_\text{aux}) +\P(\Ec_{1, 2, ..., k}^c) +k!\cdot \left(\prod_{i=1}^l q_i(\|\wv_i\|_2+\epsilon, \epsilon)\right) \cdot \sum_{{\Tc}\subset[k]} 2^{(k-|{\Tc}|)\log m}\cdot 2^{-\frac{n}{2}\log\frac{\left((1-\epsilon) \sigma_a^2\right)^l \det \left(\underline{W}_{[k]\backslash\Tc}^\intercal \underline{W}_{[k]\backslash\Tc} + \frac{\sigma_z^2}{\sigma_a^2}I \right)}{(\sigma_z^2+\epsilon^2 \sigma_a^2)^l}}\nonumber\\
&= \P(\Ec^c_\text{aux}) +\P(\Ec_{1, 2, ..., k}^c) +k!\cdot \left(\prod_{i=1}^l q_i(\|\wv_i\|_2+\epsilon, \epsilon)\right)  \cdot\sum_{{\Tc}\subseteq[k]} 2^{|{\Tc}|\log m} \cdot 2^{-\frac{n}{2} \log\frac{\left((1-\epsilon) \sigma_a^2\right)^l \det \left(\underline{W}_{\Tc}^\intercal \underline{W}_{\Tc} + \frac{\sigma_z^2}{\sigma_a^2}I \right)}{(\sigma_z^2+\epsilon^2 \sigma_a^2)^l} }\nonumber
\end{align}
which tends to zero as $m\rightarrow \infty$, if
\begin{align}
\limsup_{m\rightarrow \infty}\frac{\log m}{n_m}<\frac{1}{2|{\Tc}|} \log\frac{\left((1-\epsilon) \sigma_a^2\right)^l \det \left(\underline{W}_{\Tc}^\intercal \underline{W}_{\Tc} + \frac{\sigma_z^2}{\sigma_a^2}I \right)}{(\sigma_z^2+\epsilon^2 \sigma_a^2)^l}
\label{general_condition}
\end{align}
for all $\Tc\subseteq[k]$. Since $\epsilon > 0$ is arbitrarily chosen, the proof of Theorem 1 is complete.

\medskip
Next, we prove Lemma 1. For $j\in\Sc$, $(b_{i,j}-D_{i,j})^2/\theta_{j}$ is a noncentral $\chi^2$ random variable. Its moment generating function is \cite{Lancaster_x2} (for $t< 1/2$)
\begin{align}
\E[e^{t(b_{i,j}-D_{i,j})^2/\theta_j}] &= \frac{e^{\frac{t{b_{i,j}^2}/{\theta_j}}{1-2 t} }}{(1-2t)^{\frac{1}{2}}}.
\end{align}
By changing variable $\frac{\theta_j t}{nl}\rightarrow t$, we have
\begin{align}
\E[e^{\frac{t(b_{i,j}-D_{i,j})^2}{nl}}] &= \frac{e^{\frac{\frac{t}{nl}{b_{i,j}^2}}{1-\frac{2 \theta_j t}{nl}} }}{(1-\frac{2\theta_j t}{nl})^{\frac{1}{2}}}.
\end{align}
For $j\in[l]\backslash\Sc$ with $\Dv_j\equiv\mathbf{0}$, we additionally define $\theta_j = 0$. In this case,
\begin{align}
\E[e^{\frac{t(b_{i,j}-D_{i,j})^2}{nl}}] &=\E[e^{\frac{tb_{i,j}^2}{nl}}] =  e^{\frac{t}{nl}{b_{i,j}^2}}= \frac{e^{\frac{\frac{t}{nl}{b_{i,j}^2}}{1-\frac{2 \theta_j t}{nl}} }}{(1-\frac{2\theta_j t}{nl})^{\frac{1}{2}}}.
\end{align}

Define
\begin{align}
S_n \triangleq \frac{1}{nl}\|B-D\|_F^2 = \frac{1}{l} \sum_{j=1}^l \frac{1}{n}\|\bv_j - \Dv_j\|_2^2.
\end{align}
Then, we have
\begin{align}
\E[e^{tS_n}] &= \E[e^{ \frac{t}{nl}\|B-D\|_F^2}]\\
&= \E[e^{  \frac{t}{l} \sum_{j=1}^l \frac{1}{n}\|\bv_j - \Dv_j\|_2^2 }]\\
&= \prod_{j=1}^{l} \E[e^{  \frac{t}{nl} \|\bv_j - \Dv_j\|_2^2 }]\\
&= \prod_{j=1}^{l} \frac{e^{\frac{\frac{t}{nl}{\|\bv_j\|_2^2}}{1-\frac{2 \theta_j t}{nl}} }}{(1-\frac{2\theta_j t}{nl})^{\frac{n}{2}}}
\end{align}

The Chernoff bound indicates that
\begin{align}
\P(S_n \leq \gamma) &\leq  \min_{s>0}e^{s\gamma}\E[e^{-sS_n}]\\
&= \min_{s>0} e^{s\gamma} \prod_{j=1}^{l} \frac{e^{\frac{-\frac{s}{nl}{\|\bv_j\|_2^2}}{1+\frac{2 \theta_j s}{nl}} }}{(1+\frac{2\theta_j s}{nl})^{\frac{n}{2}}}\\
&= \min_{p<0} e^{-p\gamma} \prod_{j=1}^{l} \frac{e^{\frac{\frac{p}{nl}{\|\bv_j\|_2^2}}{1-\frac{2 \theta_j p}{nl}} }}{(1-\frac{2\theta_j p}{nl})^{\frac{n}{2}}}\\
&= \exp\left\{\min_{p<0}\left\{-p\gamma+ \sum_{j=1}^l \left[\frac{\frac{p}{nl}{ \|\bv_j\|_2^2}}{1-\frac{2 \theta_j p}{nl}}  - {\frac{n}{2}}\log\left(1-\frac{2 \theta_j p}{nl}\right)\right]  \right\}\right\}\\
&= \exp\left\{\min_{p<0}\left\{-lp\gamma+ \sum_{j=1}^l \left[\frac{\frac{p}{n}{ \|\bv_j\|_2^2}}{1-\frac{2 \theta_j p}{n}}  - {\frac{n}{2}}\log\left(1-\frac{2 \theta_j p}{n}\right)\right]  \right\}\right\}\\
&= \exp\left\{\min_{p<0}\left\{-lp\gamma - \sum_{j=1}^l \underbrace{\frac{\frac{-p}{n}{ \|\bv_j\|_2^2}}{1-\frac{2 \theta_j p}{n}}}_{>0}  -  {\frac{n}{2}}\log\prod_{j=1}^l\left(1-\frac{2 \theta_j p}{n}\right)  \right\}\right\}\\
&\leq \exp\left\{\min_{p<0}\left\{-lp\gamma - l\left(\prod_{j=1}^l \frac{\frac{-p}{n}{ \|\bv_j\|_2^2}}{1-\frac{2 \theta_j p}{n}}\right)^{\frac{1}{l}}  -  {\frac{n}{2}}\log\prod_{j=1}^l\left(1-\frac{2 \theta_j p}{n}\right)  \right\}\right\}\label{step2}\\
&= \exp\left\{\min_{p<0}\left\{-lp\gamma - l \frac{\left(\prod_{j=1}^l \frac{-p}{n}{ \|\bv_j\|_2^2}\right)^{\frac{1}{l}}}{\left(\prod_{j=1}^l(1-\frac{2 \theta_j p}{n}) \right)^{\frac{1}{l}}}  -  {\frac{nl}{2}}\log\left(\left(\prod_{j=1}^l\left(1-\frac{2 \theta_j p}{n}\right)\right)^{\frac{1}{l}} \right) \right\}\right\}\\
&= \exp\left\{\min_{p<0}\left\{\underbrace{-lp\gamma + lp \frac{\left(\prod_{j=1}^l \frac{1}{n}{ \|\bv_j\|_2^2}\right)^{\frac{1}{l}}}{\left(\prod_{j=1}^l(1-\frac{2 \theta_j p}{n}) \right)^{\frac{1}{l}}}  -  {\frac{nl}{2}}\log\left(\left(\prod_{j=1}^l\left(1-\frac{2 \theta_j p}{n}\right)\right)^{\frac{1}{l}}\right)}_{\triangleq f(p)}  \right\}\right\}\\
&= \exp\left\{\min_{p<0} f(p)\right\}.
\end{align}
where (\ref{step2}) follows from the fact that the arithmetic mean is no smaller than the geometric mean. On the other hand, define the function
\begin{align}
g(p, \theta) = -lp\gamma + lp \frac{\left(\prod_{j=1}^l \frac{1}{n}{ \|\bv_j\|_2^2}\right)^{\frac{1}{l}}}{1-\frac{2 \theta p}{n}} - {\frac{nl}{2}}\log\left(1-\frac{2 \theta p}{n}\right) .
\end{align}
Recall that $ \left(\prod_{j=1}^l \frac{1}{n}{ \|\bv_j\|_2^2}\right)^{\frac{1}{l}}= \alpha $. It can be readily seen that, for a fixed $p<0$,
\begin{align}
f(p) \leq \max_{\theta > 0}  g(p,\theta)
\end{align}
which is because there exists $\theta>0$ such that $1-\frac{2 \theta p}{n} = \left(\prod_{j=1}^l\left(1-\frac{2 \theta_j p}{n}\right)\right)^{\frac{1}{l}}$. Thus,
\begin{align}
\min_{p<0}f(p) \leq \min_{p<0}\left(\max_{\theta > 0}  g(p,\theta)\right).
\end{align}
Our goal is to show
\begin{align}
 \min_{p<0}\left(\max_{\theta > 0}  g(p,\theta)\right) = -\frac{nl}{2}\log \frac{\alpha}{\gamma}
\end{align}
which will lead to $\P(S_n \leq \gamma)\leq -\frac{nl}{2}\log \frac{\alpha}{\gamma}$ as desired. To this end, we first consider, for a fixed $p$,
\begin{align}
\frac{\partial g(p,\theta)}{\partial \theta} = \frac{pl(\frac{2p\alpha}{n}- \frac{2p\theta}{n}+1)}{(1-\frac{2\theta p}{n})^2}
\end{align}
By setting $\frac{\partial g(p,\theta)}{\partial \theta} = 0$, we have the stationary point
$
\theta^* = \alpha + \frac{n}{2p}
$.
Examine the second derivative
\begin{align}
\frac{\partial^2 g(p,\theta)}{\partial \theta^2}\Bigg|_{\theta = \theta^*} &= \frac{(1-\frac{2\theta p}{n})(-\frac{2p^2l}{n})(\frac{2p\theta}{n} - \frac{4p\alpha}{n}-1 )  }{(1-\frac{2\theta p}{n})^4}\Bigg|_{\theta = \theta^*}= \frac{(-\frac{2p^2l}{n}) }{(\frac{2\alpha p}{n})^2} < 0.
\end{align}
Due to the constraint $\theta >0$, we have
\begin{align}
\max_{\theta > 0}  g(p,\theta) = \left\{ \begin{array}{ll}
-pl\gamma -\frac{nl}{2}-\frac{nl}{2}\log(-\frac{2p\alpha}{n}) & \mbox{if $p\leq -\frac{n}{2\alpha}$};\\
pl(\alpha-\gamma) & \mbox{if $-\frac{n}{2\alpha} < p < 0 $}.\end{array} \right.
\end{align}

Next, we calculate $\min_{p<0}\left(\max_{\theta > 0}  g(p,\theta)\right)$. First,
\begin{align}
\min_{ -\frac{n}{2\alpha}<p<0}\left(\max_{\theta > 0}  g(p,\theta) \right)= \min_{ -\frac{n}{2\alpha}<p<0}pl(\alpha-\gamma)=-\frac{nl}{2}\left(1-\frac{\gamma}{\alpha}\right).
\end{align}
Then, to figure out $\min_{p\leq  -\frac{n}{2\alpha}}\left(\max_{\theta > 0}  g(p,\theta)\right)$, we compute
\begin{align}
\frac{\partial \max_{\theta > 0}  g(p,\theta)}{\partial p} &= \frac{\partial \left(-pl\gamma -\frac{nl}{2}-\frac{nl}{2}\log(-\frac{2p\alpha}{n}) \right)}{\partial p}= -l\gamma - \frac{nl}{2p} =_{\textrm{set}} 0
\end{align}
which gives the stationary point $
p^* = -\frac{n}{2\gamma}$.
Check for the second derivative,
\begin{align}
\frac{\partial^2 \max_{\theta > 0}  g(p,\theta)}{\partial p^2}\Bigg|_{p=p^*} &= \frac{nl}{2(p^*)^2}>0.
\end{align}
Therefore, $p^* = -\frac{n}{2\gamma} (\leq -\frac{n}{2\alpha})$ is the minimizer. As a result,
\begin{align}
\min_{p\leq  -\frac{n}{2\alpha}} \left(\max_{\theta > 0}  g(p,\theta) \right)= -pl\gamma -\frac{nl}{2}-\frac{nl}{2}\log(-\frac{2p\alpha}{n})\Bigg|_{p = p^*} = -\frac{nl}{2}\log\frac{\alpha}{\gamma}.
\end{align}

Overall,
\begin{align}
\min_{p<0}\left(\max_{\theta > 0}  g(p,\theta) \right)= \min\left(-\frac{nl}{2}\left(1-\frac{\gamma}{\alpha}\right) ,-\frac{nl}{2}\log\frac{\alpha}{\gamma}\right)
\end{align}
Using the fact that $0\leq 1-\frac{1}{x} \leq \log x$ for $x>1$, we finally have
\begin{align}
\min_{p<0}\left(\max_{\theta > 0}  g(p,\theta)\right) = -\frac{nl}{2}\log\frac{\alpha}{\gamma}.
\end{align}

Therefore,
\begin{align}
\P(S_n \leq \gamma) &\leq\exp\left\{\min_{p<0}f(p) \right\}\\
&\leq \min_{p<0}\left(\max_{\theta > 0}  g(p,\theta)\right)\\
&= 2^{-\frac{n}{2}\log \frac{\prod_{j=1}^l \frac{1}{n}{ \|\bv_j\|_2^2}}{\gamma^l}}\\
&= 2^{-\frac{n}{2}\log \frac{\prod_{j=1}^l { [\frac{1}{n}B^\intercal B]_{j,j}}}{\gamma^l}}.
\end{align}

The remaining task is to prove Lemma \ref{lemma3}. Let $\sigma_{b, 1}^2 \geq \cdots \geq \sigma_{b, q}^2$ be the $q$ eigenvalues of $B^\intercal B$, where $\sigma_{b, q}^2 = \sigma_{b}^2$. The eigen-decomposition states that there exists a unitary matrix $J\in\mathbb{R}^{q\times q}$, such that $B^\intercal B = J G G J^\intercal$, where $G\in\mathbb{R}^{q\times q}$ is a diagonal matrix with the $i$th diagonal element being $\sigma_{b, i}$. Thus, $ D^\intercal B^\intercal B D = D^\intercal J G G J^\intercal D = F T$,
where $F = D^\intercal J G$ and $T = F^\intercal$. Note that
\begin{align}
&\det((BD)^\intercal B D)\nonumber\\
&= \det(FT) \nonumber\\
&= \sum_{1\leq j_1<\cdots<j_r\leq q} \det\left[\begin{array}{ccc} f_{1, j_1} & \cdots & f_{1, j_r}\\ \vdots & &\vdots \\ f_{r, j_1} & \cdots & f_{r, j_r} \end{array}\right]
\det\left[\begin{array}{ccc} t_{j_1, 1} & \cdots & t_{j_1, r}\\ \vdots & &\vdots \\ t_{k_r, 1} & \cdots & t_{j_r, r} \end{array}\right]\label{lemma4_step1}\\
&= \sum_{1\leq j_1<\cdots<j_r\leq q} \left(\det\left[\begin{array}{ccc} f_{1, j_1} & \cdots & f_{1, j_r}\\ \vdots & &\vdots \\ f_{r, j_1} & \cdots & f_{r, j_r} \end{array}\right]\right)^2\nonumber\\
&= \sum_{1\leq j_1<\cdots<j_r\leq q} \left(\det\left\{\left[\begin{array}{ccc} [ D^\intercal J ]_{1, j_1} & \cdots & [ D^\intercal J ]_{1, j_r}\\ \vdots & &\vdots \\ {[ D^\intercal J ]}_{r, j_1} & \cdots & [ D^\intercal J ]_{r, j_r} \end{array}\right] \diag(\sigma_{b, j_1}, ..., \sigma_{b, j_r})\right\}  \right)^2\nonumber\\
& \geq (\sigma_b^2)^r \sum_{1\leq j_1<\cdots<j_r\leq q} \left(\det\left[\begin{array}{ccc} [ D^\intercal J ]_{1, j_1} & \cdots & [ D^\intercal J ]_{1, j_r}\\ \vdots & &\vdots \\ {[ D^\intercal J ]}_{r, j_1} & \cdots & [ D^\intercal J ]_{r, j_r} \end{array}\right] \right)^2\nonumber\\
&= (\sigma_b^2)^r \det(D^\intercal J^\intercal J D)\nonumber\\
&= (\sigma_b^2)^r \det(D^\intercal D)\nonumber
\end{align}
where (\ref{lemma4_step1}) is due to the Binet-Cauchy formula \cite{Dym}.

\section{Proof of Theorem 2}
\label{app:th2}
To establish this theorem, we prove the following equivalent statement:

If there exist a sequence of matrices $\{A^{(m)}\}_{m=k}^\infty$, $A^{(m)}\in\mathbb{R}^{n_m\times m}$, and a sequence of support recovery maps $\{d^{(m)}\}_{m=k}^\infty$, $d^{(m)}:\mathbb{R}^{n_m}\mapsto 2^{\{1, 2, ..., m\}}$, such that
\[
\frac{1}{n_m m}\|A^{(m)}\|_F^2 \leq \sigma_a^2
\]
and
\[
\lim_{m\rightarrow\infty} \P\{d^{(m)}(A^{(m)} X + Z)\neq \supp(X(W, \Sv))\} = 0
\]
then
\[
\limsup_{m\rightarrow\infty}\frac{\log m}{n_m} \leq c(W).
\]

\medskip

For any $\Tc\subseteq[k]$, denote the tuple of random
variables $(S_l:l\in\Tc)$ by $S(\Tc)$.  For notation simplicity, let $\overline{P}_e^{(m)} \triangleq \P\{d^{(m)}(A^{(m)} X + Z)\neq \supp(X(W, \Sv))\} $.  From Fano's inequality
\cite{Cover06}, we have
\begin{align}
H(S(\Tc)|Y) &\leq H(S_1, ..., S_k|Y) \nonumber\\
&\leq \log k!+ H(\{S_1, ..., S_k\}|Y)\nonumber\\
& \leq \log k! + \overline{P}_e \log{m\choose k} + 1.
\label{T2_fanos}
\end{align}
On the other hand,
\begin{align}
H(S(\Tc)|S(\Tc^c))
&=\log\left(\prod\limits_{q=0}^{|\Tc|-1}(m-(k-|\Tc|)-q)\right)\nonumber\\
&= |\Tc|\log m - n\epsilon_{1,n}
\label{general_starting_point}
\end{align}
where $\Tc^c \triangleq [k] \backslash \Tc$ and
\[
\epsilon_{1,n} \triangleq  \frac{1}{n} \log
\left({m^{|\Tc|}}/{\prod\limits_{q=0}^{|\Tc|-1}(m-(k-|\Tc|)-q)}\right)\nonumber
\]
which tends to zero as $n\rightarrow \infty$.  Hence, combining
(\ref{T2_fanos}) and (\ref{general_starting_point}), we have
\begin{align}
|\Tc|\log m &= H(S(\Tc)|S(\Tc^c))+n\epsilon_{1,n}\nonumber\\
&= I(S(\Tc);Y|S(\Tc^c)) + H(S(\Tc)|Y,S(\Tc^c))+n\epsilon_{1,n}\nonumber\\
&\leq  I(S(\Tc);Y|S(\Tc^c)) + H(S(\Tc)|Y)+n\epsilon_{1,n}\label{explain6}\\
&\leq  I(S(\Tc);Y|S(\Tc^c)) + \log k! + \overline{P}_e^{(m)} \log{m\choose k} +1+ n\epsilon_{1,n} \nonumber\\
&= \sum_{i=1}^n I(\underline{\Yv}_i;S(\Tc)|\underline{Y}_{[i-1]}, S(\Tc^c))+ \log k! + \overline{P}_e^{(m)} \log{m\choose k} + 1+ n\epsilon_{1,n}\nonumber\\
&= \sum_{i=1}^n \left(h(\underline{\Yv}_i|\underline{Y}_{[i-1]}, S(\Tc^c)) - h(\underline{\Yv}_i|\underline{Y}_{[i-1]}, S([k]))\right) + \log k! + \overline{P}_e^{(m)} \log{m\choose k} + 1+ n\epsilon_{1,n}\nonumber\\
&\leq \sum_{i=1}^n \left(h(\underline{\Yv}_i|S(\Tc^c))- h(\underline{\Yv}_i|S_1,...,S_k)\right)+ \log k! + \overline{P}_e^{(m)} \log{m\choose k} + 1+ n\epsilon_{1,n}\label{explain7}\\
&= \sum_{i=1}^n \left(h(\underline{\Yv}_i|S(\Tc^c))- h(\underline{\Zv}_i)\right)+ \log k! + \overline{P}_e^{(m)} \log{m\choose k} + 1+ n\epsilon_{1,n}\label{reduce_to_noise}
\end{align}
where $\underline{Y}_{[i-1]}$ denotes the set $\{\underline{\Yv}_1, ..., \underline{\Yv}_{i-1}\}$. To explain some intermediate steps, (\ref{explain6}) follows from the fact that conditioning reduces entropy, (\ref{explain7}) holds because $\underline{\Yv}_i$ is independent of $\underline{Y}_{[i-1]}$ when conditioned on $S([k])$, and (\ref{reduce_to_noise}) follows since the measurement matrix is fixed
and $\underline{\Zv}_i$ is independent of $(S_1,\ldots, S_k)$.

Consider
\begin{align}
&h(\underline{\Yv}_i|S(\Tc^c))\nonumber\\
&= h\left(A_{i, S([k])} W + \underline{\Zv}_i \Big|S(\Tc^c)\right)  \nonumber\\
&= h\left(A_{i, S(\Tc)}\underline{W}_{\Tc} + \underline{\Zv}_i \Big|S(\Tc^c)\right)  \nonumber\\
&\leq  h\left(A_{i, S(\Tc)}\underline{W}_{\Tc} + \underline{\Zv}_i\right) \nonumber\\
&\leq  \frac{1}{2}\log\left((2\pi e)^l\cdot  \det\left(\E[(A_{i, S(\Tc)}\underline{W}_{\Tc} + \underline{\Zv}_i)^\intercal (A_{i, S(\Tc)}\underline{W}_{\Tc} + \underline{\Zv}_i)] - \E[A_{i, S(\Tc)}\underline{W}_{\Tc} + \underline{\Zv}_i]^\intercal \E[A_{i, S(\Tc)}\underline{W}_{\Tc} + \underline{\Zv}_i]\right) \right)\label{explain5}\\
&\leq \frac{1}{2}\log\left((2\pi e)^l\cdot  \det\left(\underline{W}_{\Tc}^\intercal \left(\E[ A_{i, S(\Tc)}^\intercal A_{i, S(\Tc)} ]- \E[A_{i, S(\Tc)}]^\intercal \E[A_{i, S(\Tc)}]\right)\underline{W}_{\Tc} + \sigma_z^2 I\right) \right)
\label{property_different_entropy}
\end{align}
where (\ref{explain5}) follows from the fact that with the same covariance the Gaussian random vector maximizes the entropy \cite{Cover06}, and the randomness in $A_{i, S(\Tc)}$ is due to the randomness of the index set $S(\Tc)$. Note that
\begin{align}
\E[A_{i, S(\Tc)}] = \frac{1}{m}\sum_{p=1}^m a_{i, p} \mathbf{1}^\intercal.
\end{align}
Meanwhile
\begin{align}
\E[A_{i, S(\Tc)}^\intercal A_{i, S(\Tc)}] = \frac{1}{m}\sum_{p=1}^m a_{i, p}^2 I + \frac{1}{m(m-1)} \sum_{p=1}^m\sum_{\substack{q=1\\q\neq p}}^m a_{i,p}a_{i,q}  (\mathbf{1}\cdot \mathbf{1}^\intercal - I).
\end{align}
Thus
\begin{align}
&\E[A_{i, S(\Tc)}^\intercal A_{i, S(\Tc)}] - \E[A_{i, S(\Tc)}]^\intercal \E[A_{i, S(\Tc)}]\nonumber\\
&= \frac{1}{m}\sum_{p=1}^m a_{i, p}^2 I + \frac{1}{m(m-1)} \sum_{p=1}^m\sum_{\substack{q=1\\q\neq p}}^m a_{i,p}a_{i,q}  (\mathbf{1}\cdot \mathbf{1}^\intercal - I) - \frac{1}{m^2}\left(\sum_{p=1}^m a_{i, p}\right)^2 \mathbf{1}\cdot\mathbf{1}^\intercal\nonumber\\
&= \frac{1}{m}\sum_{p=1}^m a_{i, p}^2 I + \frac{1}{m(m-1)} \left(\left(\sum_{p=1}^m a_{i, p}\right)^2 - \sum_{p=1}^m a_{i, p}^2\right)  (\mathbf{1}\cdot \mathbf{1}^\intercal - I) - \frac{1}{m^2}\left(\sum_{p=1}^m a_{i, p}\right)^2 \mathbf{1}\cdot\mathbf{1}^\intercal\nonumber\\
&= \frac{1}{m}\sum_{p=1}^m a_{i, p}^2 \left(I - \frac{1}{m-1}(\mathbf{1}\cdot \mathbf{1}^\intercal - I)\right) + \left(\sum_{p=1}^m a_{i, p}\right)^2 \left( \frac{1}{m(m-1)}(\mathbf{1}\cdot \mathbf{1}^\intercal - I) - \frac{1}{m^2} \mathbf{1}\cdot \mathbf{1}^\intercal   \right).
\end{align}
Note that $\frac{1}{m(m-1)}(\mathbf{1}\cdot \mathbf{1}^\intercal - I) - \frac{1}{m^2} \mathbf{1}\cdot \mathbf{1}^\intercal = \frac{1}{m^2(m-1)}\mathbf{1}\cdot \mathbf{1}^\intercal - \frac{1}{m(m-1)} I $ is negative semidefinite for sufficiently large $m$, and so is $\underline{W}_{\Tc}^\intercal\left(\frac{1}{m^2(m-1)}\mathbf{1}\cdot \mathbf{1}^\intercal - \frac{1}{m(m-1)} I\right)\underline{W}_{\Tc}$. Hence
\begin{align}
&\det\left(\underline{W}_{\Tc}^\intercal \left(\E[ A_{i, S(\Tc)}^\intercal A_{i, S(\Tc)} ]- \E[A_{i, S(\Tc)}]^\intercal \E[A_{i, S(\Tc)}]\right)\underline{W}_{\Tc} + \sigma_z^2 I\right)\nonumber\\
&\leq \det\left( \frac{1}{m}\sum_{p=1}^m a_{i, p}^2 \underline{W}_{\Tc}^\intercal \left(I - \frac{1}{m-1}(\mathbf{1}\cdot \mathbf{1}^\intercal - I)\right) \underline{W}_{\Tc}  + \sigma_z^2 I\right)\nonumber
\end{align}
as a result of \cite[Corollary 8.4.15]{matrixmath}. Therefore
\begin{align}
&|\Tc|\log m \nonumber\\
&\leq \sum_{i=1}^n \left[\frac{1}{2}\log\left((2\pi e)^l\cdot  \det\left(\frac{1}{m}\sum_{p=1}^m a_{i, p}^2\underline{W}_{\Tc}^\intercal \left(I - \frac{1}{m-1}(\mathbf{1}\cdot \mathbf{1}^\intercal - I)\right) \underline{W}_{\Tc}  + \sigma_z^2 I\right)\right)  - \frac{1}{2}\log\left((2\pi e \sigma_z^2)^l \right)\right]\nonumber\\
&~~~~~+ \log k! + \overline{P}_e^{(m)} \log{m\choose k} + 1+ n\epsilon_{1,n}\nonumber\\
&= \sum_{i=1}^n \frac{1}{2}\log\det\left(\frac{1}{m\sigma_z^2}\sum_{p=1}^m a_{i, p}^2 \underline{W}_{\Tc}^\intercal \left(I - \frac{1}{m-1}(\mathbf{1}\cdot \mathbf{1}^\intercal - I)\right) \underline{W}_{\Tc}  + I\right) \nonumber\\
 &~~~~~+ \log k! + \overline{P}_e^{(m)} \log{m\choose k} + 1+ n\epsilon_{1,n}\nonumber\\
&\leq \frac{n}{2}\log\det\left(\frac{1}{nm\sigma_z^2}\sum_{i=1}^n\sum_{p=1}^m a_{i, p}^2  \underline{W}_{\Tc}^\intercal \left(I - \frac{1}{m-1}(\mathbf{1}\cdot \mathbf{1}^\intercal - I)\right) \underline{W}_{\Tc} +  I\right) \nonumber\\
 &~~~~~ + \log k! + \overline{P}_e^{(m)} \log{m\choose k} + 1+ n\epsilon_{1,n}\nonumber\\
&\leq \frac{n}{2}\log\det\left(\frac{\sigma_a^2}{\sigma_z^2} \underline{W}_{\Tc}^\intercal \left(I - \frac{1}{m-1}(\mathbf{1}\cdot \mathbf{1}^\intercal - I)\right) \underline{W}_{\Tc}  +  I\right) + \log k! + \overline{P}_e^{(m)} \log{m\choose k} + 1+ n\epsilon_{1,n}\nonumber\\
&\leq \frac{n}{2}\log\det\left(\frac{\sigma_a^2}{\sigma_z^2} \underline{W}_{\Tc}^\intercal \left(I - \frac{1}{m-1}(\mathbf{1}\cdot \mathbf{1}^\intercal - I)\right) \underline{W}_{\Tc}  +  I\right) + \log k! + \overline{P}_e^{(m)} k\log m + 1+ n\epsilon_{1,n}.
\end{align}
Then, we have
\begin{align}
&\limsup_{m\rightarrow\infty} \frac{(1-k\overline{P}_e^{(m)}/|\Tc|)\log m}{n_m} - \frac{\log k! + n_m \epsilon_{1,n} + 1}{|\Tc|n_m}\nonumber\\
&\leq \limsup_{m\rightarrow\infty}  \frac{1}{2|\Tc|}\log\det\left(\frac{\sigma_a^2}{\sigma_z^2} \underline{W}_{\Tc}^\intercal \left(I - \frac{1}{m-1}(\mathbf{1}\cdot \mathbf{1}^\intercal - I)\right) \underline{W}_{\Tc}  +  I\right)\nonumber\\
&=  \frac{1}{2|\Tc|}\log\det\left(\frac{\sigma_a^2}{\sigma_z^2} \underline{W}_{\Tc}^\intercal \underline{W}_{\Tc}  +  I\right)
\end{align}
for all $\Tc\subseteq [k]$. Since $\lim_{m\rightarrow\infty} \overline{P}_e^{(m)} = 0$, we reach the conclusion
\begin{align}
\limsup_{m\rightarrow\infty} \frac{\log m}{n_m}
\leq\frac{1}{2|\Tc|}\log\det\left(\frac{\sigma_a^2}{\sigma_z^2} \underline{W}_{\Tc}^\intercal \underline{W}_{\Tc}  +  I\right)
\end{align}
for all $\Tc\subseteq [k]$. This completes the proof of Theorem 2.


\section{Proof of Corollary \ref{corollary3}}
\label{A_ProofofExample}
To justify this corollary, we need to show
\begin{align}
\min_{\Tc\subseteq [k]}\left[\frac{1}{2|\Tc|}\log\det\left(I + \frac{\sigma_a^2}{\sigma_z^2}{\underline{W}_{\Tc}}^\intercal {\underline{W}}_{\Tc}\right)  \right] = 2\cdot\frac{1}{2k} \log\left(1+k\cdot\frac{\sigma_a^2}{\sigma_z^2}\right).\nonumber
\end{align}

To begin with, recall that $k$ is even, and $\wv_2$ is defined in (\ref{definew2}). For a given $\Tc\subseteq[k]$, let $\Tc_1 = \Tc\cap [\frac{k}{2}]$, $\Tc_2 = \Tc \backslash\Tc_1$, $t = |\Tc|$, $t_1 = |\Tc_1|$, and $t_2 = |\Tc_2|$. One can obtain
\begin{align}
{\underline{W}_\Tc}^\intercal \underline{W}_\Tc = \left[\begin{array}{cc}
t & t_1 - t_2 \\
t_1 - t_2 & t \end{array}\right].\nonumber
\end{align}
Let  $\alpha \triangleq \frac{\sigma_a^2}{\sigma_z^2}$ for notational simplicity. Thus
\begin{align}
\frac{1}{2|\Tc|}\log\det\left(I + \alpha{\underline{W}_{\Tc}}^\intercal {\underline{W}}_{\Tc}\right) &= \frac{1}{2t}\log\det\left[\begin{array}{cc}
1+\alpha t & \alpha(t_1 - t_2) \\
\alpha(t_1 - t_2) & 1+\alpha t \end{array}\right]\nonumber\\
&= \frac{1}{2t}\log \left( 1+2\alpha t + 4 \alpha^2 t_1t_2   \right)
\end{align}
where we use the fact that $t=t_1 + t_2$. Note that, for a given $t\in[k]$,
\begin{align}
\min_{\Tc: \Tc\subseteq [k], |\Tc| = t \leq \frac{k}{2}} \frac{1}{2t}\log \left( 1+2\alpha t + 4 \alpha^2 t_1t_2   \right) = \frac{1}{2t}\log ( 1+2\alpha t)
\end{align}
and
\begin{align}
\min_{\Tc: \Tc\subseteq [k],|\Tc| = t > \frac{k}{2}} \frac{1}{2t}\log \left( 1+2\alpha t + 4 \alpha^2 t_1t_2   \right) = \frac{1}{2t}\log \left( 1+2\alpha t + 4 \alpha^2 \frac{k}{2}\left(t- \frac{k}{2}\right)   \right)
\end{align}
where we use the implicit constraints that $t_1, t_2 \leq \frac{k}{2}$. Then, the problem becomes evaluating
\begin{align}
\min_{t:t\in[k]} f(t), \textrm{ where }
f(t) = \left\{ \begin{array}{ll}
\frac{1}{2t}\log ( 1+2\alpha t) & \textrm{if }0<t\leq \frac{k}{2},\\
\frac{1}{2t}\log \left( 1+2\alpha t + 4 \alpha^2 \frac{k}{2}\left(t- \frac{k}{2}\right)   \right) & \textrm{if }\frac{k}{2}+1\leq t\leq k.\end{array} \right.
\end{align}
First, it can be readily seen that $\min_{t:t\in[\frac{k}{2}]} f(t) = \frac{1}{k}\log(1+\alpha k)$. Next, we consider the function
\[
g(t) \triangleq \frac{\log(\beta_1 + \beta_2 t)}{2t}
\]
where $\beta_1 \triangleq 1-\alpha^2k^2$ and $\beta_2 \triangleq 2\alpha(1+\alpha k)$ for $t\in[\frac{k}{2}, k]$. Note that\footnote{For the purpose of analysis, the base of logarithm is not important, as long as all of them are consistent. Here, we choose natural logarithm to simplify the calculation.}
\begin{align}
\frac{\partial g(t)}{\partial t} = \frac{1-\frac{\beta_1}{\beta_1+\beta_2 t}-\log (\beta_1+\beta_2 t)}{t^2}.
\end{align}
To obtain stationary points, we solve
\begin{align}
1-\frac{\beta_1}{\beta_1+\beta_2 t}+\log \frac{1}{\beta_1+\beta_2 t} = 0, ~t\neq 0
\end{align}
which is equivalent to
\begin{align}
1+v(t) = \beta_1 e^{v(t)}, ~t\neq 0
\label{stationary}
\end{align}
where $v(t) \triangleq \log \frac{1}{\beta_1+\beta_2 t}$. Note that $\beta_1 < 1$. We will consider three different cases. The first case is $0<\beta_1<1$. By comparing the curves of $1+v$ and $\beta_1 e^{v}$ as functions of $v$, we see that there are two solutions with opposite signs, namely $v_1<0$ and $v_2>0$, to (\ref{stationary}). Note that
\[
g\left(\frac{k}{2}\right) = g(k) = \frac{1}{k}\log(1+\alpha k)\nonumber.
\]
Meanwhile, $v(t)$ is monotonically decreasing on $[\frac{k}{2}, k]$, and
\[
v(k) = \log \frac{1}{(1+\alpha k)^2} <v\left(\frac{k}{2}\right) = \log \frac{1}{1+\alpha k} <0.
\]
Therefore, it is evident that $v(k) < v_1< v\left(\frac{k}{2}\right)<v_2$. Further, it can be readily seen that
\begin{align}
\frac{\partial g(t)}{\partial t}\Bigg|_{t=\frac{k}{2}} = \frac{1+v(t) - \beta_1 e^{v(t)}}{t^2}\Bigg|_{t=\frac{k}{2}} >0\nonumber\\
\frac{\partial g(t)}{\partial t}\Bigg|_{t=k} = \frac{1+v(t) - \beta_1 e^{v(t)}}{t^2}\Bigg|_{t=k} <0\nonumber.
\end{align}
In summary, $g(t)$ is increasing at $t=\frac{k}{2}$ and decreasing at $t=k$, it takes the same value at these two points, and there exists only one stationary point in between. These observations lead to the conclusion that $\min_{t:t\in[k]\backslash[\frac{k}{2}]} f(t) =f(k)= \frac{1}{k}\log(1+\alpha k)$.

To analyze the cases for $\beta_1 = 0$ and $\beta_1 < 0$, we only need to note that there is only one solution $v_1$ to (\ref{stationary}). Thus, similar argument applies to these two cases.

\newpage
%
%
%
%

\end{document}